\documentclass[prx,twocolumn,showpacs,amsmath,amssymb,superscriptaddress,floatfix,reprint, nofootinbib]{revtex4-2}
\usepackage{amsmath}
\usepackage{amssymb}
\usepackage{amsthm}
\usepackage{amsfonts}
\usepackage{dsfont}
\usepackage{listings}
\usepackage{macros}
\usepackage{enumerate}
\usepackage{latexsym}
\usepackage{pifont}

\usepackage{cancel}
\usepackage{hyperref}
\newtheorem{lemma}{Lemma}

\usepackage{bbold}

\usepackage{psfrag}
\usepackage{xcolor}
\usepackage{float}

\usepackage{comment}

\usepackage{graphicx}
\usepackage{subfigure}

\makeatletter
\def\l@subsubsection#1#2{}
\makeatother

\begin{document}

\tolerance 10000
\title{Strong Zero Modes via Commutant Algebras}
\author{Sanjay Moudgalya}
\email{sanjay.moudgalya@gmail.com}
\affiliation{School of Natural Sciences, Technische Universit\"{a}t M\"{u}nchen (TUM), James-Franck-Str. 1, 85748 Garching, Germany}
\affiliation{Munich Center for Quantum Science and Technology (MCQST), Schellingstr. 4, 80799 M\"{u}nchen, Germany}
\author{Olexei I. Motrunich}
\email{motrunch@caltech.edu}
\affiliation{Department of Physics and Institute for Quantum Information and Matter,
California Institute of Technology, Pasadena, California 91125, USA}
\begin{abstract}
Strong Zero Modes (SZMs) are (approximately) conserved quantities that result in (approximate) double degeneracies in the entire spectra of certain Hamiltonians, with the Majorana zero mode of the transverse-field Ising chain being a primary example. In this work, we discover via a systematic search that many examples of SZMs can be understood as symmetries in the commutant algebra framework, which reveals novel algebraic structures hidden in Hamiltonians with well-known SZMs, including the transverse-field Ising chain. Our findings unify the understanding of different examples of SZMs in the literature, demystify their connections to ground state phases of matter, and reveal novel symmetries in simple models, such as exact quasilocal $U(1)$ symmetries that sometimes accompany the SZMs such as in the spin-1/2 XY model for certain parameter values. Moreover, while analytically tractable SZMs have mostly been demonstrated only for non-interacting or integrable models, the algebraic structures revealed in this work can be exploited to construct integrability-breaking interactions that exactly preserve these SZMs. Such non-integrable models are expected to show more clear dynamical signatures of SZMs without the interference of other conserved quantities that appear in integrable models, and we discuss many examples, including those of novel hydrodynamic modes associated with such symmetries for some special parameter values. We also show that while this commutant understanding extends to the non-interacting limit of the celebrated Fendley SZM in the spin-1/2 XYZ chain, the SZM in the interacting case cannot be understood in this framework. This suggests that there are two types of SZMs -- those that survive integrability breaking and those that do not. We close by using this commutant understanding to construct an alternate proof of the Fendley SZM, which might be of independent interest. 
\end{abstract}
\date{\today}
\maketitle
%


\tableofcontents

\section{Introduction}
\label{sec:intro}
Quantum many-body Hamiltonians can exhibit several kinds of symmetries. 
Much of the quantum many-body physics literature implicitly assumes some kinds of restrictions imposed on symmetry operators, e.g., to on-site unitary symmetries, or certain lattice symmetries.
This is for good reason, since these appear to be the ones that play the most important roles in explaining numerous physical phenomena of interest, and they are also the easiest to detect and analyze.
The very definition of symmetry beyond these simple ones is non-trivial, e.g., the vast generalization that says \textit{any} operator that commutes with the Hamiltonian is a symmetry has a major caveat, since the exponentially many projectors onto eigenstates of the Hamiltonian satisfy that definition, which, however, are not very meaningful symmetries of the problem.
Nevertheless, numerous recent works have uncovered novel definitions of symmetries and examples of natural Hamiltonians that realize them.
These include ``categorical symmetries," studied in the context of equilibrium physics~\cite{ji2020categorical, mcgreevy2022generalized, shao2023noninvertible, bhardwaj2024lectures, bhardwaj2024lattice}, where symmetry operators can be non-invertible and need not form simple group structures.
In the context of non-equilibrium physics, unconventional symmetries occur in systems exhibiting weak ergodicity breaking~\cite{serbyn2020review, papic2021review, moudgalya2021review, chandran2022review}, particularly Hilbert Space Fragmentation (HSF)~\cite{sala2020fragmentation, khemani2020localization, moudgalya2019thermalization, yang2019hilbertspace} and Quantum Many-Body Scars (QMBS)~\cite{moudgalya2021review, chandran2022review}, where the symmetry operators can be non-local and again need not have any simple group structures~\cite{moudgalya2021hilbert, moudgalya2022exhaustive}.  
One way of systematically defining these unconventional symmetries is with the use of so-called \textit{commutant algebras}, which widely appear in the quantum information literature~\cite{zanardi2001virtual, bartlett2007reference}, and have been used to characterize symmetries that lead to weak ergodicity breaking~\cite{moudgalya2021hilbert, moudgalya2022exhaustive}.
In this framework, symmetry algebras are understood as centralizers or commutants of \textit{bond algebras}~\cite{nussinov2009bond, moudgalya2021hilbert, chatterjee2022algebra}, which are algebras generated by some set of strictly local operators that appear in the Hamiltonian.
This allows the exhaustive construction of Hamiltonians that exhibit those symmetries~\cite{moudgalya2022from, moudgalya2022exhaustive}, leads to numerical methods to detect these symmetries~\cite{moudgalya2023numerical},  and also to the understanding of hydrodynamics due to such symmetries~\cite{moudgalya2023symmetries}.
The fact that many conventional and unconventional symmetries are understood in this framework begs the question of whether other symmetries in the literature can also be understood in this manner, and if there may exist also yet undiscovered interesting symmetries.
While we defer the question of exploring the full landscape of symmetries to a separate work, in this paper we focus on an interesting class of symmetries that we find via a systematic search, known as Strong Zero Modes (SZMs).
SZMs are exponentially localized conserved quantities that appear at the edge of one-dimensional (1d) chains in the thermodynamic limit.
Many such conserved quantities have been studied in the recent literature under names such as exponential or modulated symmetries~\cite{sala2022dynamics, lehmann2023fragmentation, hu2023spontaneous, lian2023quantum, hu2024bosonic} or in the context of Hilbert space fragmentation~\cite{rakovszky2020statistical, moudgalya2021hilbert, buca2022dynamical, buca2023unified, li2025dynamics}.
However, in this work we will use SZM to refer to its historical definition~\cite{fendley2016strong} (reviewed in Sec.~\ref{sec:SZM}), in that they are akin to Majorana Zero Modes in the Kitaev chain~\cite{pfeuty1970ising, kitaev2001unpaired, motrunich2001griffiths} which anticommute with a $\mathbb{Z}_2$ symmetry of the system and lead to double degeneracies of the \textit{entire spectrum} in the thermodynamic limit.
Such SZMs have been observed in numerous models in both Hamiltonian and Floquet settings.
Conceptually simple realizations occur in non-interacting or effectively free-fermion systems~\cite{munk2018dyonic, vasiloiu2018enhancing, vasiloiu2019strong, santos2020strong,
monthus2018even, mcginley2019slow, bjerlin2022probing, dag2020topologically, mahyaeh2020study}, but they have also been obtained exactly in a variety of interacting integrable systems~\cite{fendley2012parafermionic, fendley2016strong, chepiga2017exact, moran2017parafermionic,
maceira2018infinite, pellegrino2020constructing, kemp2020symmetry, olund2023boundary, yates2021strong,
yeh2023slowly} with a notable example being the Fendley SZM in the Bethe Ansatz integrable spin-1/2 XYZ chain~\cite{fendley2016strong}.
Potential applications of SZMs as stable qubits~\cite{alicea2016topological, else2017prethermal} for quantum computation have been proposed, and they have also been realized experimentally~\cite{mi2022noise}.
SZMs are to be contrasted with ``weak'' zero modes, which are similar to edge Majorana modes in topological phases and lead to a double degeneracy of only the ground state (and low-energy excitations) in the thermodynamic limit, representing a more generic scenario of topological phases~\cite{koma2022stability}.
Exact examples of weak zero modes have been found in both non-interacting and interacting non-integrable spin chains
\cite{katsura2015exact, kawabata2017exact, wouters2018exact, fu2022exact, mahyaeh2018exact}.
In contrast, SZMs, as historically defined, were long thought to occur exactly only in non-interacting or interacting \textit{integrable} models, with integrability-breaking perturbations believed to eventually wash out their signatures, although remnants of SZM physics can survive for long times even in such settings~\cite{yates2019almost, yates2020dynamics, wada2021coexistence, yates2022longlived, mukherjee2023emergent, pandey2023out, laflorencie2022topological, laflorencie2023universal}.
In this work, we resolve this confusion by explicitly constructing non-integrable models that exhibit SZMs, clearly demonstrating for the first time that integrability is not a necessary condition for their existence.
To do so, we first show that several examples of SZMs can also be understood in the commutant algebra framework, which itself demystifies many interesting features they possess.
First, unlike usual conserved quantities, the SZMs studied in the literature usually do not exactly commute with their respective Hamiltonians at finite system sizes, but only up to terms exponentially small in the system size. 
Second, they are conserved \textit{only} in certain ground state phases of their Hamiltonians, e.g., in the ferromagnetic or topological phases~\cite{katsura2015exact, fendley2016strong}, where they also have consequences for the infinite-temperature dynamics and excited states, and they cease to be conserved in the other phases.
This correlation between dynamical infinite-temperature properties and zero-temperature ground state properties is an unexpected feature at first sight. 
We clearly illustrate the origin of these features by revisiting standard examples of SZMs, including the transverse field Ising model, in the new light of their understanding in the commutant framework.
This also naturally paves the way for the construction of non-integrable models with these SZMs.
In addition, it also reveals novel quasilocal $U(1)$ symmetries that can sometimes accompany these SZMs, which lead to interesting dynamical signatures that we discuss.
Finally, we also revisit the celebrated Fendley SZM in the spin-1/2 XYZ model~\cite{fendley2016strong}, an interacting Bethe Ansatz integrable model, and show that apart from some particular choices of parameters, it remains outside the commutant framework, making it different from other examples of SZMs.
Nevertheless, in our attempt to understand the connection of the Fendley SZM to commutants, we discover a simplified proof of its existence in the spin-1/2 XYZ model, which might be of independent interest.
This paper is organized as follows.
In Sec.~\ref{sec:commutants}, we review key concepts required for this work, which are commutant algebras, numerical methods to obtain them, and precise working definitions of SZMs.
In Sec.~\ref{sec:Ising}, we discuss the SZM of the transverse-field Ising model in the language of commutant algebras, and demonstrate connections to ground state phases of matter. 
Then in Sec.~\ref{sec:XYSZM}, we discuss a generalized version of the SZM in the spin-1/2 XY model, where we also find examples of novel quasilocal $U(1)$ symmetries that accompany the SZMs along some lines in the parameter space.
In Sec.~\ref{subsec:nonintSZM}, we use this commutant understanding to explicitly construct non-integrable models with the Ising and XY SZMs, and discuss their dynamical implications, including hydrodynamics with novel conserved quantities. 
Finally, in Sec.~\ref{sec:fendleySZM} we revisit the celebrated Fendley SZM in the spin-1/2 XYZ model, show the commutant understanding in its non-interacting limit, and discuss its differences from the other examples of SZMs.
Finally, we conclude with open questions in Sec.~\ref{sec:conclusions}.
\section{Review of Key Concepts}\label{sec:commutants}
\subsection{Symmetries and Commutant Algebras}
We first briefly review key concepts of commutant algebras that we need for this work; more comprehensive discussions can be found in our previous papers on this subject~\cite{moudgalya2021hilbert, moudgalya2022from, moudgalya2022exhaustive, moudgalya2023numerical, moudgalya2023symmetries}.
The core idea is to think of symmetries in terms of a pair of operator algebras $(\mA, \mC)$, referred to as the \textit{bond algebra} and \textit{commutant algebra} or the \textit{symmetry algebra} respectively.
The bond algebra $\mA$ is the associative algebra generated by taking products and linear combinations of a set of strictly local Hermitian operators $\{\hH_\alpha\}$ on a lattice or a spin chain, and we denote it as $\mA = \lgen \{\hH_\alpha\} \rgen$.\footnote{The notion of a bond algebra can also be generalized to include sums of local operators (also referred to as extensive local operators) in its generating set $\{\hH_\alpha\}$~\cite{moudgalya2022from, moudgalya2022exhaustive}, however we will not need such generalization for this work.} 
(Note that it is implicit that we include $\mathds{1}$ in $\mathcal{A}$.)
The commutant algebra $\mC$ is defined as the set of all operators that commute with the $\{\hH_\alpha\}$:
\begin{equation}
    \mC = \{\hO: [\hO, \hH_\alpha] = 0,\;\;\forall\;\alpha\}.
\label{eq:commutantdefn}
\end{equation}
This set forms an algebra since it is closed under products and linear combinations, and it enjoys the additional property that it is closed under Hermitian conjugation as well.
$\mC$ is the symmetry algebra (encompassing all conserved quantities) for \textit{all} families of Hamiltonians that can be expressed in terms of linear combinations of products of $\{\hH_\alpha\}$.
Thinking of symmetries in this framework has led to a systematic understanding of phenomena such as Hilbert space fragmentation~\cite{moudgalya2021hilbert}, quantum many-body scars~\cite{moudgalya2022exhaustive}, and also regular symmetries~\cite{moudgalya2022from} and their hydrodynamic modes~\cite{moudgalya2023symmetries}. 
The bond and commutant algebras are closed under Hermitian conjugation, and hence for finite-dimensional Hilbert spaces they are examples of von Neumann algebras~\cite{landsman1998lecture, harlow2017}. 
This leads to several nice properties, and symmetry sectors (corresponding to the symmetry algebra $\mC$) for the Hamiltonians constructed out of the generators of $\mA$ can be understood in terms of the representation theory of von Neumann algebras.
In particular, this guarantees that there is a unitary transformation $W$ that simultaneously brings the algebras into the following block-diagonal matrix forms~\cite{zanardi2001virtual, bartlett2007reference, lidar2014dfs, moudgalya2021hilbert, moudgalya2022from, moudgalya2022exhaustive, fulton2013representation, moudgalya2023numerical}:

\begin{align}
    &W^\dagger \hh_{\mA} W = \bigoplus_{\lambda} (M^\lambda(\hh_\mA) \otimes \mathds{1}_{d_\lambda}),\nn\\
    &W^\dagger \hh_{\mC} W = \bigoplus_{\lambda} (\mathds{1}_{D_\lambda} \otimes N^\lambda(\hh_\mC)), \label{eq:ACblock}
\end{align}
where $M^\lambda(\hh_\mA)$ and $N^\lambda(\hh_\mC)$ are matrices of dimensions $D_\lambda$ and $d_\lambda$ respectively, which are dimensions of irreducible representations of $\mA$ and $\mC$ respectively.
Hence any Hamiltonian in the bond algebra $\mathcal{A}$ (for which the symmetries are $\mathcal{C}$) would have a spectrum that has levels that are $d_\lambda$-fold degenerate for all $\lambda$.
For example, when $\mC$ is (the associative algebra of the generators of) $SU(2)$, $\lambda$ denotes the $\vec{S}^2_{\rm tot}$ eigenvalue $\lambda(\lambda + 1)$ and $d_\lambda = 2\lambda + 1$, specifying the degeneracy due to different values of $S^z_{\rm tot}$.
The number of linearly independent operators in $\mathcal{C}$ is referred to as its \textit{dimension} $\dim(\mC)$, which is also given by  $\dim(\mC) = \sum_\lambda d_{\lambda}^2$.
There are also various numerical methods to construct the commutant $\mC$ given the generators $\{\hH_\alpha\}$~\cite{moudgalya2023numerical}, which we use for the verification of results presented here. 
\subsection{Systematic search for symmetries}
Before we present our main results, we briefly outline a procedure for searching for novel symmetries that we use to find bond algebras with novel kinds of symmetries. 
In particular, we scan through bond algebras of the form $\mA(\vmu) \defn \lgen \{\hH_\alpha(\vmu) \} \rgen$ with generators that depend on a few parameters that we denote as $\vmu$, and identify the corresponding commutant algebras $\mC(\vmu)$. 
This can in principle be computed by constructing the super-Hamiltonian $\hmP(\vmu)$ corresponding to the generators $\{\hH_\alpha(\vmu)\}$, and determining its ground state manifold; a convenient quantity to compute is $\dim(\mC(\vmu))$.
For most values of parameters, we expect a ``generic" commutant, depending on the symmetries of the parameter space we are searching over. 
We restricted our search space to qubit systems, in particular to nearest-neighbor generators
\begin{gather}
    \hH_\alpha(\vec{\mu}) \defn (\kappa + \gamma)X_j X_{j+1} + (\kappa - \gamma) Y_j Y_{j+1} \nn \\
    + h_+ (Z_j + Z_{j+1}) + h_- (Z_j - Z_{j+1}),
\label{eq:Hfamily}
\end{gather}
where $\vec{\mu} = (\kappa, \gamma, h_+, h_-)$, which are all parameters with a $\mathbb{Z}_2$ symmetry $Q_Z$ given by
\begin{equation}
    Q_Z = \prod_j{Z_j}.
\label{eq:Qzdefinition}
\end{equation}
This family includes a large class of physically relevant models, such as the transverse-field Ising model, spin-1/2 XY and XYZ models, and the Heisenberg model.
For generic values of parameters we expect $\dim(\mC(\vmu)) = 2$, corresponding to the commutant composed only of the $\mathbb{Z}_2$ symmetry (i.e., $\mC = \{\mathds{1}, Q_Z\})$. 
We further ensure that at least one of $\kappa$ and $\gamma$ is non-zero, so there is at least one ``non-classical" term in our search, which is essential to obtain an interesting non-trivial commutant.
We perform the search for both open boundary conditions (OBC) and periodic boundary conditions (PBC) on a chain of $L$ sites labeled $1, \dots, L$, but all the results on SZMs we discuss in this work are for OBC.
Amidst these generic commutants, there are special parameter values $\vmu = \vmu_0$ where $\dim(\mC(\vmu = \vmu_0))$ is larger, which are points or lower-dimensional manifolds in the parameter space that possess additional symmetries.
To detect such commutants, we use the fact that the commutant can be viewed as the ground state of a local superoperator $\hmP(\vmu)$~\cite{moudgalya2023numerical, moudgalya2023symmetries}, defined as
\begin{equation}
    \hmP(\vmu) \defn \sum_{\alpha}{(\hH_\alpha(\vmu) \otimes \mathds{1} - \mathds{1} \otimes \hH_\alpha(\vmu)^T)^2}.
\end{equation}
We then study the low-energy spectrum of $\hmP(\vmu)$ and compute a  \textit{commutant entropy} $S(\vmu)$, defined as
\begin{equation}
    S(\vmu) \defn - \sumal{\nu}{}{p_\nu \log p_\nu},\;\;p_\nu \defn \frac{\exp(- \beta \epsilon_\nu)}{\sum_\nu{\exp(-\beta \epsilon_\nu)}},
\label{eq:commentropy}
\end{equation}
where $\beta$ is a parameter that we choose as suitable for the problem at hand, and $\{\epsilon_\nu\geq 0\}$ are the eigenvalues of $\hmP(\vmu)$; and we have suppressed the explicit dependence of the quantities on $\beta$ and $\vmu$.
Heuristically, as $\vmu$ approaches such a special point $\vmu = \vmu_0$ with a larger commutant, we expect the density of states of $\hmP(\mu)$ close to the ground state to increase, since $\hmP(\mu_0)$ has a larger number of ground states by definition. 
Since $S(\vmu)$ has the interpretation as the thermal entropy of a system with inverse temperature $\beta$ and energy eigenvalues $\{\epsilon_\nu\}$, it is expected to increase smoothly as $\vmu \rightarrow \vmu_0$, allowing us to detect the presence of special commutants with discrete scans over parameter space. 
We find that the landscape of commutants obtained in this way is intimidatingly large. 
Rather than try to exhaustively characterize all possible commutants, here we choose to focus on some simple interesting examples that resulted from this search, with SZMs being prime examples.
\subsection{Strong Zero Modes (SZM)}\label{sec:SZM}
We now review the definition of the SZMs that we will stick to in this work~\cite{alicea2016topological, fendley2016strong}, which might differ slightly from those used in different parts of the literature. 
Considering a one-dimensional chain of $L$ sites, we say that a Hermitian operator $\Psi$ is an \textit{exact} SZM if it satisfies the following properties:
\begin{enumerate}
\item It is an exact conserved quantity, i.e., $[\Psi, H] = 0$ for Hamiltonians $H$ in the bond algebra, and such $\Psi$ exists for any finite system size $L$. 
\item It is \textit{exponentially} localized at one of the boundary of the chain.
In particular, the weight of the operator $\Psi$ on $x$ sites closest to the boundary (as measured by the weight of the operator basis strings with support on those sites) scales as $1 - e^{-x/\xi}$, where $\xi$ is a characteristic lengthscale. 
\item It satisfies $\Psi^2 = \mathcal{N}_L \mathds{1}$, where $\mathcal{N}_L$ is a normalization factor potentially dependent on the system size $L$.
    This condition, while not fundamentally necessary, makes $\Psi$ similar to a Majorana Zero Mode (MZM), and also immediately gives the operator norm of the SZM as $\norm{\Psi} = \sqrt{\mathcal{N_L}}$.
    \item It anticommutes with the $\mathbb{Z}_2$ symmetry of the system, i.e., we have $\{\Psi, Q_Z\} = 0$, which gives rise to exact double degeneracy of all eigenstates of the Hamiltonian $H$.
\end{enumerate}
Note that we want to distinguish this from \textit{approximate} SZMs (which are usually referred to as just SZMs in the literature~\cite{alicea2016topological, fendley2016strong}), which differ in two aspects. 
First, the property 1 is replaced by $\norm{[\Psi, H]} \sim e^{-c L}$, i.e., $\Psi$ need not be an exact conserved quantity at finite system sizes.  
Second, the property 3 is required to hold in the thermodynamic limit, with $\mathcal{N}_L$ converging to a value independent of the system size, i.e., $\Psi$ should be \textit{normalizable} in the thermodynamic limit. 
We will mostly focus on exact SZMs in the discussion below, but we also show that the Hamiltonian can be slightly perturbed to make them approximate SZMs that coincide with those studied in the literature.
Further, since we study $\mathbb{Z}_2$ symmetric spin-1/2 Hamiltonians, many of the SZMs localized on the ends of the chain are closely related to Majorana Zero Modes (MZMs), which are sums of single Majorana operators, localized on the ends of the chain after a Jordan-Wigner transformation, as reviewed for many examples in App.~\ref{app:jordanwigner}.
However, strictly speaking, these are not always in a one-to-one correspondence with each other, i.e., a localized SZM does not always map to a localized MZM in the fermionic language due to Jordan-Wigner strings that can go across the system.
Nevertheless, in the presence of a $\mathbb{Z_2}$ symmetry, the symmetric operator can be multiplied to obtain a localized MZM as well, hence we will often be using SZM and MZM interchangeably even though there is this slight technicality.
There is also yet another distinct concept of \textit{almost} SZM~\cite{yates2019almost, yates2020lifetime, olund2023boundary, laflorencie2022topological, tausendpfund2025almost}, which are obtained when perturbations are added to models with approximate SZMs as defined above, and have a finite but small system-size independent commutator with the Hamiltonian.
In addition, there are many examples of symmetry operators that are exponentially localized at the ends of the system~\cite{rakovszky2020statistical, sala2022dynamics, lehmann2023fragmentation, hu2023spontaneous, lian2023quantum, hu2024bosonic}, i.e., those that only satisfy conditions 1 and 2 above, which are sometimes known as exponential or modulated symmetries.
We will not be discussing either of these in this work.
\section{Ising Strong Zero Mode}\label{sec:Ising}
We start with a simple example, which relates to the case of the SZM in the Transverse-Field Ising model. 
We start with the bond algebra of the form 
\begin{equation}
    \mA_{\iszm} \defn \lgen \{X_j X_{j+1} + q Z_j\} \rgen,
\label{eq:Isingbond}
\end{equation}
with $1 \leq j \leq L-1$, which corresponds to $(\kappa, \gamma, h_+, h_-) = (1, 1, q, q)$ in Eq.~(\ref{eq:Hfamily}).
We numerically find that the commutant of this bond algebra has $\dim(\mC_{\iszm}) = 8$, which we can understand as described below.
Note that a very similar analysis can be carried out for $\gamma = \kappa$, $h_- = -h_+$ or $\gamma = -\kappa$, $h_- = \pm h_+$, which all have the same Ising-like forms.
\subsection{Structure of the Commutant}
To determine the commutant, we observe that the following relations are satisfied for the generators $T_{j,j+1}$ of $\mA_{\iszm}$ in Eq.~(\ref{eq:Isingbond}):
\begin{equation}
    [X_j + q Z_j X_{j+1}, T_{j,j+1}] = 0.
\label{eq:Isinglocal}
\end{equation}
By inspection, this allows us to build the conserved quantity
\begin{equation}
    \Psi^\ell(q) \defn \sumal{j = 1}{L}{q^{j-1} \left(\prodal{i = 1}{j-1}{Z_i}\right) X_j}. 
\label{eq:ISZMdefns}
\end{equation}
In addition, for OBC, we also observe that $X_L$ commutes with all of the generators, since the rightmost site $L$ does not have a $Z_L$ in any of the generators. 
Remembering that the generators also have an overall $\mathbb{Z}_2$ symmetry generated by $Q_Z$ of Eq.~(\ref{eq:Qzdefinition}), the full commutant of $\mA_{\rm I-SZM}$ is given by
\begin{align}
    \mC_{\iszm} &= \lgen Q_Z, \Psi^\ell(q), X_L \rgen \nn \\
    &=\text{span}\{\mathds{1},\; Q_Z,\; \Psi^\ell(q),\; X_L,\; Q_Z \Psi^\ell(q),\; Q_Z X_L,\;\nn\\
    &\Psi^\ell(q) X_L,\; Q_Z \Psi^\ell(q) X_L\},
\label{eq:ISZMcomm}
\end{align}
where we have used the fact that $[\Psi^\ell(q)]^2 \propto \mathds{1}$, $Q_Z^2 = X_L^2 = 1$,  $Q_Z \Psi^\ell(q) = -\Psi^\ell(q) Q_Z$, $Q_Z X_L = -X_L Q_Z$, and $\Psi^\ell(q) X_L = X_L \Psi^\ell(q)$.
Note that $\Psi^\ell(q)$ is an example of an \textit{exact} SZM for all Hamiltonians built out of $\mA_{\iszm}$ since it satisfies all the properties mentioned in Sec.~\ref{sec:SZM} for $|q| < 1$.
It is exactly conserved for any finite system size and is exponentially localized on the left boundary of the chain for $|q| < 1$ (but is delocalized if $|q| > 1$).
Moreover, it satisfies $[\Psi^\ell(q)]^2 = \frac{1-|q|^{2L}}{1 - |q|^2}\mathds{1} \propto \mathds{1}$, and finally it anticommutes with the $\mathbb{Z}_2$ symmetry, i.e., $\{\Psi^\ell(q), Q_Z\} = 0$, which gives rise to the exact double degeneracy of all eigenstates of such Hamiltonians. 
Note that the commutant $\mC_{\iszm}$ also has other operators that satisfy the required properties of SZM, albeit for $|q| > 1$.
This includes operators such as
\begin{equation}
    \widetilde{\Psi}^\ell(q) \defn -i q^{-L+1}Q_Z \Psi^\ell(q) =   \sum_{j =1}^L{(q^{-1})^{L-j} Y_j \left(\prod_{i = j+1}^{L}Z_i\right)},
\label{eq:Psir}
\end{equation}
which is an SZM on the right boundary if $|q| > 1$, but delocalized if $|q| < 1$. 
In addition, there is $X_L$, which is strictly speaking an SZM, but less important since it does not have any impact on the physics of the Ising model, which we discuss below. 
Indeed, we can always add the generators $\{X_j Y_{j+1}, j=1,\dots, L-1\}$ to the bond algebra of Eq.~(\ref{eq:Isingbond}) to break this $X_L$ symmetry;\footnote{However, note that for any linear combination of these generators, a different SZM that depends on the precise couplings in the Hamiltonian will always be approximately conserved.
This follows from the non-interacting structure of such Hamiltonians, see Sec.~\ref{subsubsec:SPspectrum} for a more detailed discussion.} we then obtain the following pair of bond and commutant algebras
\begin{align}
    \widetilde{\mA}_{\iszm} &= \lgen \{X_j X_{j+1} + q Z_j\}, \{X_j Y_{j+1}\} \rgen,\nn \\
    \widetilde{\mC}_{\iszm} &= \lgen Q_Z, \Psi^\ell(q) \rgen, 
\label{eq:IsingXLbroken}
\end{align}
which has $\dim(\widetilde{\mC}_{\iszm}) = 4$.
Note that these SZMs are no longer symmetries under PBC, i.e., they do not commute with the generator $T_{L,1}$ that has support on sites $L$ and $1$, hence the commutant for PBC is simply the regular $\mathbb{Z}_2$ symmetry. 
However, note that it is plausible that the SZMs can survive under appropriately twisted boundary conditions~\cite{kawabata2017exact}, a scenario which we do not explore further here. 
\subsection{Connection to phases of the Ising Model}\label{subsec:phaseconnection}
We now connect the SZM $\Psi^\ell(q)$ to the Majorana mode of the Transverse-Field Ising Model (TFIM), given by
\begin{equation}
    H_{\tfim} = -\sumal{j = 1}{L-1}{X_j X_{j+1}} - q \sumal{j = 1}{L}{Z_j} = \widetilde{H}_{\tfim} - q Z_L,
\label{eq:TFIMhamil}
\end{equation}
where we have defined
\begin{equation}
    \widetilde{H}_{\tfim} = -\sumal{j = 1}{L-1}{(X_j X_{j+1}} + q Z_j),
\label{eq:TFIMtild}
\end{equation}
which is just a sum of generators of Eq.~(\ref{eq:Isingbond}).
It is often said that the SZM only exists when the ground state of the TFIM is in its ferromagnetic phase (i.e., $|q| < 1$), where it is essentially the Majorana mode in the non-interacting language.
This might be surprising given the analysis of the commutant of $\mA_{\iszm}$, which is valid for any $q$. 
Below we clarify this relation to ground state phases in two ways, first by analyzing its single-particle spectrum, and second, by considering the effect of the edge field $q Z_L$ in Eq.~(\ref{eq:TFIMhamil}).
\subsubsection{Single-particle spectrum}\label{subsubsec:SPspectrum}
First, note that the TFIM, and the generators of the algebra $\mA_{\iszm}$ [and also $\widetilde{\mA}_{\iszm}$] are all non-interacting terms after a Jordan-Wigner transformation, see App.~\ref{app:jordanwigner} for a review.
That is, any such term $T$ when written in terms of Majorana fermions $\{\chi_k\}_{k= 1}^{2L}$ is of the form $T = \sum_{m,n}{M_{m,n} \chi_m \chi_n}$, where $M$ is a $2L \times 2L$ antisymmetric Hermitian matrix, which we refer to as the \textit{single-particle} Hamiltonian. 
For the generators of $\mA_{\iszm}$, we have (see App.~\ref{app:jordanwigner} for definitions and derivations)
\begin{equation}
    X_j X_{j+1} + q Z_j = i \chi_{2j} \chi_{2j+1} + i q \chi_{2j-1}\chi_{2j}.
\label{eq:Agen}
\end{equation}
In these non-interacting models, we know that MZMs, which are zero-energy eigenstates (or eigenstates with an exponentially small energy for a finite system size) of the single-particle Hamiltonian $M$ always appear in pairs, which is a consequence of $M$ being an even-dimensional antisymmetric matrix. 
It is easy to show that these MZMs that have close to zero-energy are conserved quantities of the full Hamiltonian (more precisely, for any single-particle mode $\Psi = \sum_m{a_m \chi_m}$ that has single-particle energy $\varepsilon$, the commutator with the Hamiltonian $\norm{[H, \Psi]} = 4\varepsilon \norm{\Psi}$).
In the commutant $\mC_{\iszm}$ of Eq.~(\ref{eq:ISZMcomm}), we can identify two operators that are linear in the Majorana fermions:
\begin{equation}
    \Psi^{\ell}(q) = \sum_{j = 1}^L{q^{j-1} \chi_{2j-1}},\;\;\;
    i Q_Z X_L = \chi_{2L}.
\label{eq:SZMMajorana}
\end{equation}
These operators would necessarily appear as a pair of MZMs of $\widetilde{H}_{\tfim}$, which is in $\mA_{\iszm}$, and are generically the \textit{only} pair of zero-energy modes of this Hamiltonian. 
While this is true for any value of $q$, these MZMs are localized on opposite ends of the chain \textit{only} if $|q| < 1$, i.e., in the ferromagnetic phase of the model. 
The two-fold degeneracy of the ground state caused by the SZMs is then ``stable" only when the corresponding MZMs are localized on opposite ends of the chain and do not hybridize under local perturbations. 
Equivalently, the two-fold degeneracy of the ground state is due to the ``bulk" spontaneous symmetry breaking (or equivalently, topological order of the fermion superconductor) only when $|q| < 1$, and it cannot be removed by local perturbations in that regime.
On the other hand, the degeneracy when $|q| > 1$ is simply due to the boundary, and can easily be removed by adding a local perturbation.
In this sense, while the SZM does exist for any value of $q$, it is stable only in the ferromagnetic phase.
\subsubsection{Approximate SZMs}\label{subsubsec:approxSZM}
The stability physics discussed above becomes even more clear when we study the effect of adding back the field $q Z_L$ in Eq.~(\ref{eq:TFIMhamil}), which we now analyze.
Since $\widetilde{H}_{\tfim}$ of Eq.~(\ref{eq:TFIMtild}) commutes with the SZM $\Psi^\ell(q)$, its commutator with $H_{\tfim}$ of Eq.~(\ref{eq:TFIMhamil}) reads
\begin{equation}
[H_{\tfim}, \Psi^\ell(q)] = -2 i q^L Z_1 Z_2 \dots Z_{L-1} Y_L ~.
\end{equation}
The operator norm for the commutator of the Hamiltonian with the normalized SZM operator is then
\begin{equation}
\frac{\norm{[H_{\tfim}, \Psi^\ell(q)]}}{\norm{\Psi^\ell(q)}} = 2q^L \sqrt{\frac{1-q^2}{1-q^{2L}}} ~.
\label{eq:approxSZMcondition}
\end{equation}
The norm is hence exponentially small if $|q| < 1$, i.e., in the ferromagnetic phase, whereas it is $O(1)$ when $|q| > 1$, i.e., in the paramagnetic phase.
Hence $\Psi^\ell(q)$ is an approximate SZM of $H_{\tfim}$ \textit{only} in its ferromagnetic phase.\footnote{\label{ft:approxSZM2}Note that $H_{\tfim}$ in Eq.~(\ref{eq:TFIMhamil}) is inversion symmetric, hence we can similarly construct an approximate SZM localized near the right boundary starting from the bond algebra $\lgen \{X_j X_{j+1} + q Z_{j+1}\rgen$.
However, note that this additional approximate SZM does not lead to any extra degeneracies in $H_{\tfim}$ since it commutes with $\Psi^\ell(q)$ and anticommutes with $Q_Z$.}
This approximate SZM also leads to an approximate two-fold degeneracy of the entire spectrum for $|q| < 1$ assumed below.
Indeed, if $\ket{\phi}$ is a normalized eigenstate of $H_{\tfim}$ with a definite $Q_Z$ quantum number and energy $\varepsilon$, then $\ket{\phi'} = \Psi^\ell(q) \ket{\phi}/\norm{\Psi^\ell(q)}$ is a normalized state with opposite $Q_Z$ (hence orthogonal to $\ket{\phi}$) and satisfies $\norm{H_{\tfim} \ket{\phi'} - \varepsilon\ket{\phi'}} < 2|q|^L$.
In particular, this shows the two-fold degeneracy for the (gapped) ground states with an exponentially small splitting, as well as for the low-energy excitations.
Moreover, since the many-body energy splitting is $\mathcal{O}(2^{-L})$, for sufficiently small $|q| < 1/2$, we can guarantee that we also get an approximate two-fold degeneracy (with energy levels being closer than the typical level spacing) everywhere in the spectrum.
However, note that in the clean Ising model, this two-fold degeneracy higher up in the spectrum does not necessarily translate to ferromagnetic order everywhere in the spectrum, which would be expected only in the presence of disorder~\cite{fisher1992ising, huse2013localization}.
Since $\Psi^\ell(q)$ is an approximate symmetry in the ferromagnetic phase, and $Q_Z$ is an exact symmetry, it follows that the operator $\widetilde{\Psi}^\ell(q)$, Eq.~(\ref{eq:Psir}), is also an approximate symmetry of $H_{\tfim}$ in the ferromagnetic phase, but is not an SZM since it is delocalized in this phase.
On the other hand, other elements that span the commutant in Eq.~(\ref{eq:ISZMcomm}), such as all the terms that involve $X_L$, do not have small commutators with $H_{\tfim}$, and hence are not even approximate symmetries in either of the phases of $H_{\tfim}$.
However, there are other approximate symmetries of $H_{\tfim}$ in the ferromagnetic phase that are not captured by the commutant $\mC_{\iszm}$.
These originate from the ``other" MZM $iQ_Z X_L$ that is partner to $\Psi^\ell(q)$ in the Majorana language as discussed in Sec.~\ref{subsubsec:SPspectrum}, which is also related to the approximate SZM discussed in Ft.~\ref{ft:approxSZM2}.
While there are aspects of the non-interacting models that $\mC_{\iszm}$ does not capture by itself (but which can be fixed by additional considerations like in the above discussion), the commutant language turns out to be ideal in analyzing the dynamics of ``generic" models in $\mA_{\iszm}$, which are actually interacting non-integrable models constructed using not just linear combinations but also products of its generators. 
We discuss the physics of these in more detail in Sec.~\ref{subsec:nonintSZM}, and the structure of $\mC_{\iszm}$ there captures many aspects of the dynamics of these systems.
\section{XY Strong Zero Modes and (Hidden) \texorpdfstring{$U(1)$}{} Symmetries}\label{sec:XYSZM}
We now go beyond the simplest Ising case, and consider various bond algebras generated by the spin-1/2 XY terms.
These exhibit generalizations of the Ising commutant of $\mC_{\iszm}$ studied in the previous section, and also exhibit cases with (hidden) $U(1)$ symmetries, which we will discuss in this section.
\subsection{XY Strong Zero Mode}\label{subsec:XYSZM}
\subsubsection{Definition}
First, we find that by choosing $h_+^2 + \gamma^2 = h_-^2 + \kappa^2$ for $h_+, h_-, \gamma,\kappa \neq 0$ in Eq.~(\ref{eq:Hfamily}), we generically obtain a commutant of the same size as the $\mC_{\iszm}$, which suggest the existence of a similar SZM.
To analyze this case, we choose the parametrization $h_{\pm} = [q (\kappa + \gamma) \pm q^{-1} (\kappa - \gamma)]/2$, and study the bond algebra
\begin{gather}
    \mA_{\xyszm} = \lgen \{(\kappa + \gamma) X_j X_{j+1} + (\kappa - \gamma) Y_j Y_{j+1} \nn \\
    + q (\kappa + \gamma) Z_j + q^{-1} (\kappa - \gamma) Z_{j+1} \}\rgen, 
\label{eq:XYbond}
\end{gather}
which reduces to $\mA_{\iszm}$ of Eq.~(\ref{eq:Isingbond}) when $\kappa = \gamma$.
For general $\kappa, \gamma$, the generators $T_{j,j+1}$ of this bond algebra still satisfy Eq.~(\ref{eq:Isinglocal}), and hence $\Psi^{\ell}(q)$ in Eq.~(\ref{eq:ISZMdefns}) remains a conserved quantity.
Furthermore, we locally obtain the relation
\begin{equation}
    [q^{-1}\frac{\kappa - \gamma}{\kappa + \gamma} X_{j} Z_{j+1} +  X_{j+1}, T_{j,j+1}] = 0,
\label{eq:XYlocal}
\end{equation}
which allows us to construct the following conserved quantity $\Psi^r(s = q^{-1}\frac{\kappa - \gamma}{\kappa + \gamma})$, where we have defined
\begin{equation}
    \Psi^r(s) \defn \sum_{j = 1}^L{s^{L - j} X_j \left(\prod_{i = j+1}^{L}{Z_i}\right)}, 
\label{eq:PsiR}
\end{equation}
which is localized on the right boundary of the chain for $|s| < 1$, and delocalized when $|s| > 1$.
Note that this reduces to $X_L$ in the Ising limit, i.e., when $\kappa = \gamma$.
The commutant of $\mA_{\xyszm}$ is given by
\begin{equation}
    \mC_{\xyszm} = \lgen Q_Z, \Psi^\ell(q), \Psi^r(s = q^{-1}\frac{\kappa - \gamma}{\kappa + \gamma}) \rgen.
\label{eq:XYSZMcomm}
\end{equation}
As in the Ising case, we have $\{Q_Z, \Psi^\ell(q)\} = 0$, $\{Q_Z, \Psi^r(s)\} = 0$, and $[\Psi^\ell(q), \Psi^r(s)] = 0$, which gives rise to a double degeneracy of the spectrum and accounts for $\dim(\mC_{\xyszm}) = 8$, similar to Eq.~(\ref{eq:ISZMcomm}). 
Further, similar to Eq.~(\ref{eq:Psir}), it is possible to construct SZMs such as $Q_Z \Psi^\ell(q)$ and $Q_Z \Psi^r(s)$ that are localized for $|q| > 1$ and $|s| > 1$ or delocalized for $|q| < 1$ and $|s| < 1$ respectively. 
\subsubsection{Connections to the XY model}\label{subsubsec:XYconnections}
These SZMs appear naturally in the XY model of the form 
\begin{align}
    \widetilde{H}_{\rm XY} = &-\sum_{j=1}^{L-1}{[(\kappa + \gamma) X_j X_{j+1} + (\kappa-\gamma)Y_j Y_{j+1}]} \nn \\
    &- [q(\kappa+\gamma) + q^{-1}(\kappa - \gamma)] \sum_{j=2}^{L-1}{Z_j} \nn \\
    &-q(\kappa + \gamma) Z_1 - q^{-1}(\kappa - \gamma) Z_L, 
\label{eq:XYHamil}
\end{align}
which is a sum of the generators of $\mA_{\xyszm}$ in Eq.~(\ref{eq:XYbond}).
Since $\widetilde{H}_{\rm XY}$ is a non-interacting model in the Majorana language (see App.~\ref{app:jordanwigner} for a review), it is easy to see that the two MZM, which always appear in pairs (see discussion in Sec.~\ref{subsubsec:SPspectrum}), are $\Psi^\ell(q)$ of Eq.~(\ref{eq:SZMMajorana}), and $i Q_Z\Psi^r(s = q^{-1}\frac{\kappa - \gamma}{\kappa + \gamma})$, which reads
\begin{equation}
    i Q_Z\Psi^r(s)  = \sum_{j = 1}^L{s^{L-j}\chi_{2j}},
\label{eq:Psirmajorana}
\end{equation}
which are both linear in the Majorana fermions. 
This connection was also noted in \cite{wouters2018exact}, where they studied a version of the Kitaev chain, which maps to $\widetilde{H}_{\rm XY}$ under a Jordan-Wigner transform.
It is easy to show that these MZMs are localized on opposite sides of the chain if and only if (see App.~\ref{subsec:XYSZM})
\begin{equation}
    |q(\kappa + \gamma) + q^{-1}(\kappa - \gamma)| < 2|\kappa|.
\label{eq:FMcondition}
\end{equation}
This is precisely when the ground state of the bulk of the chain in $\widetilde{H}_{\rm XY}$ is in its ferromagnetic phase~\cite{dutta2010quantum, sachdev2011quantum}.
Hence, these modes are ``stable" only in the ferromagnetic phase, similar to discussion in Sec.~\ref{subsubsec:SPspectrum}.
The status of these SZMs change in the XY model with uniform magnetic fields on all sites, given by
\begin{align}
    H_{\rm XY} = &-\sum_{j=1}^{L-1}{[(\kappa + \gamma) X_j X_{j+1} + (\kappa-\gamma)Y_j Y_{j+1}]} \nn \\
    &- [q(\kappa+\gamma) + q^{-1}(\kappa - \gamma)] \sum_{j=1}^{L}{Z_j}. 
\label{eq:XYHamilunifrom}
\end{align}
Unlike in the case of the Ising model, neither $\Psi^\ell(q)$ nor $\Psi^r(s = q^{-1}\frac{\kappa - \gamma}{\kappa + \gamma})$ are approximate SZMs of $H_{\rm XY}$ since extra fields need to be added on \textit{both} ends of the chain to go from $\widetilde{H}_{\rm XY}$ to $H_{\rm XY}$, which yields an $\mathcal{O}(1)$ commutator of $H_{\text{XY}}$ with both the SZMs for any value of $q$, $\kappa$, $\gamma$. 
While the existence of a pair of MZMs is guaranteed in the ferromagnetic phase due to the non-interacting nature of the model (similar to the discussion in Sec.~\ref{subsubsec:SPspectrum}), in this case they are generally not related to the SZMs $\Psi^{\ell}(q)$ and $\Psi^{r}(s)$. 
\subsection{Conventional \texorpdfstring{$U(1)$}{} Symmetries} \label{subsec:conventionalU1}
We now move on to the discussion of special points in parameter space where we observe the bond algebra $\mA_{\rm XY-SZM}$ of Eq.~(\ref{eq:XYbond}) to have a larger commutant than $\mC_{\rm XY-SZM}$ of Eq.~(\ref{eq:XYSZMcomm}).
\subsubsection{Regular \texorpdfstring{$U(1)$}{}}
\label{subsubsec:regularU1}
A simple example is when we set $\gamma = 0$, we obtain the bond algebra (up to an overall factor of $\kappa$, which we can set to $1$ w.l.o.g.)
\begin{equation}
    \mA_{\xxszm} = \lgen \{X_j X_{j+1} + Y_j Y_{j+1} + q Z_j + q^{-1} Z_{j+1} \} \rgen.
\label{eq:XXbondalgebra}
\end{equation}
This is due to an additional $U(1)$ symmetry generated by $\sum_j{Z_j}$, hence the commutant of $\mA_{\xxszm}$ (with OBC) is 
\begin{equation}
    \mC_{\xxszm} = \lgen \sum_j{Z_j}, \Psi^\ell(q), \Psi^r(q^{-1}) \rgen.
\label{eq:XXcomm}
\end{equation}
Note that the algebra generated by the total spin $\sum_j{Z_j}$ includes the $\mathbb{Z}_2$ symmetry generated by $Q_Z$~\cite{moudgalya2021hilbert, moudgalya2022from}.
We observe that the dimension of the algebra is $4L$ rather than the $4(L+1)$ from the naive expectation of the basis of $\mC_{\text{XX-SZM}}$ that is of the form $(\sum_j{Z_j})^k (\Psi^\ell)^\alpha (\Psi^r)^\beta$, where $0 \leq k \leq L$, $\alpha,\beta \in \{0, 1\}$, due to some linear relations among these basis elements. 
Here we do not attempt to fully resolve such relations. 
It is easy to see that the SZMs here [more precisely, the Majorana zero modes derived from these SZMs] are always localized on the same side of the chain [i.e., Eq.~(\ref{eq:FMcondition}) is never satisfied] for all $|q| \neq 1$, and are delocalized for $|q| = 1$. 
As far as connections to generic XY-type models model are concerned, this implies that they are not ``stable" to the addition of generic perturbations, and do not play a role in the ground state phases of matter, similar to the discussion in Sec.~\ref{subsubsec:SPspectrum}.
Indeed, if we construct XY models of the form of Eqs.~(\ref{eq:XYHamil}) or (\ref{eq:XYHamilunifrom}) with $(\kappa, \gamma) = (1, 0)$, it is easy to show that the bulk of the system is in its gapped (paramagnetic) phase for generic values of $|q| \neq 1$.
The double degeneracy in the spectrum of $\widetilde{H}_{\text{XY}}$ of Eq.~(\ref{eq:XYHamil}) caused by the SZM can then be shown to be the paramagnetic product state and a state with a single non-trivial magnon that has the same energy as the paramagnet due to the boundary fields.
The case when $|q| = 1$ is special in that it falls on the critical line between the gapped and gapless phases of the XY model, but there too these SZMs are not stable to the addition of perturbations, or to imposing a uniform magnetic field on all sites. 
\subsubsection{Staggered \texorpdfstring{$U(1)$}{}}
Analogous to the regular $U(1)$ symmetry, we can also obtain a staggered $U(1)$ symmetry generated by $\sum_j{(-1)^j Z_j}$ by setting $\kappa = 0$ and then w.l.o.g.\ $\gamma = 1$ in Eq.~(\ref{eq:XYbond}), where we obtain the bond algebra
\begin{equation}
\widetilde{\mA}_{\xxszm} = \lgen \{X_j X_{j+1} - Y_j Y_{j+1} + q Z_j - q^{-1} Z_{j+1} \} \rgen,
\end{equation}
where the commutant is given by 
\begin{equation}
    \widetilde{\mC}_{\text{XX-SZM}} = \lgen \{ \sum_j{(-1)^j Z_j}, \Psi^\ell(q), \Psi^r(-q^{-1})\rgen.
\label{eq:XXtildcomm}
\end{equation}
Since we are working with OBC, this can be related to the previous regular $U(1)$ in Sec.~\ref{subsubsec:regularU1} by the following on-site unitary rotation (stated as a change of variables):
\begin{align}
& (X,Y,Z)_{4n+1} = (\tilde{X},\tilde{Y},\tilde{Z}) ~, \nn \\
& (X,Y,Z)_{4n+2} = (\tilde{X},-\tilde{Y},-\tilde{Z}) ~, \nn \\
& 
(X,Y,Z)_{4n+3} = (-\tilde{X},-\tilde{Y},\tilde{Z}) ~, \nn \\
& 
(X,Y,Z)_{4n+4} = (-\tilde{X},\tilde{Y},-\tilde{Z}) ~.
\label{eq:stagunitary}
\end{align}
Equivalently, the unitary is
\begin{equation}
U = \prod_{n=0}^{L/4-1} \left(I_{4n+1} X_{4n+2} Z_{4n+3} Y_{4n+4} \right),
\label{eq:stagunitary2}
\end{equation}
where the truncation to system size $L$ is implicit if $L$ is not a multiple of $4$.
The specific unitary simultaneously transforms {\it all} bond generators of $\widetilde{\mA}_{\xxszm}$ to those of $\mA_{\xxszm}$, up to unimportant signs.
Hence the physics of models with this staggered $U(1)$ symmetry can be directly mapped onto those with a regular $U(1)$ symmetry.
\subsection{Quasi-Local \texorpdfstring{$U(1)$}{} Symmetries}\label{subsec:quasilocalU1I}
We numerically observe that the commutant $\mC$ of the bond algebra $\mA_{\xyszm}$ of Eq.~(\ref{eq:XYbond}) has dimension $\text{dim}(\mC) = 4L$ in two other cases with either of the following parameters (i) $q (\kappa + \gamma) = q^{-1}(\kappa - \gamma)$, (ii) $q (\kappa + \gamma) = -q^{-1}(\kappa - \gamma)$.
This suggests the existence of a $U(1)$ symmetry similar to the bond algebra of $\mA_{\xxszm}$ of Eq.~(\ref{eq:XXbondalgebra}). 
We analyze these two cases separately below, and find that they can be understood as quasi-local $U(1)$ symmetries, which are sums of exponentially localized terms rather than strictly local terms. 
We will refer to them as ``regular" and ``staggered" quasi-local $U(1)$'s respectively, since in appropriate limits they reduce to the regular and staggered conventional $U(1)$ symmetries discussed in the previous subsection.
\subsubsection{Regular quasi-local \texorpdfstring{$U(1)$}{}}
\label{subsubsec:regularquasilocalU1}
In the first case, we can w.l.o.g. set $q(\kappa+\gamma) = q^{-1}(\kappa - \gamma) = 1$ [or equivalently $(\kappa, \gamma) = (\frac{q + q^{-1}}{2}, -\frac{q - q^{-1}}{2})$], which yields the following bond algebra (up to overall factors)
\begin{equation}
    \mA_{U(1), I} = \lgen \{q^{-1} X_j X_{j+1} + q Y_j Y_{j+1} + Z_j + Z_{j+1} \} \rgen. 
\label{eq:U1bondI}
\end{equation}
This is not $U(1)$-symmetric in any obvious way, unlike $\mA_{\xxszm}$ of Eq.~(\ref{eq:XXbondalgebra}).
Nevertheless, since this is a special case of the XY bond algebra of Sec.~\ref{subsec:XYSZM}, we find the conserved quantities $\Psi^\ell(q)$ of Eq.~(\ref{eq:ISZMdefns}) and $\Psi^r(s = q)$ from Eq.~(\ref{eq:PsiR}).
Using a combination of numerical methods in \cite{moudgalya2023numerical} and brute-force algebra for small system sizes, we find the following additional symmetry operators:
\begin{align}
    \mathcal{Q}^X_I(q) 
    &= \sum_{j=1}^L{Z_j} - \eta \sum_{1 \leq j_1 < j_2 \leq L}{X_{j_1} \left(\prod_{k = j_1+1}^{j_2 - 1} q Z_k\right) X_{j_2}},\nn \\
    \mathcal{Q}^Y_I(q) 
    &= \sum_{j=1}^L{Z_j} + \eta \sum_{1 \leq j_1 < j_2 \leq L}{Y_{j_1} \left(\prod_{k = j_1+1}^{j_2 - 1} q^{-1} Z_k\right) Y_{j_2}},
\label{eq:QXsymop}
\end{align}
where $\eta \defn q - q^{-1}$. 
Note that in the $q = 1$ limit (i.e., when $\gamma = 0$), these just reduce to the regular $U(1)$ symmetry discussed in Sec.~\ref{subsubsec:regularU1}.
For $|q| < 1$ (resp. $|q| > 1$), $\mQ^X_I$ (resp. $\mQ^Y_I$) is a quasi-local operator since the range-$r$ Pauli strings in its expression are exponentially suppressed as $q^r$ (resp. $q^{-r}$), making it a sum of exponentially localized terms.
This quasi-local $U(1)$ can loosely be interpreted as a ``deformation" of the regular $U(1)$ symmetry, although the $U(1)$ structure is far from obvious; and indeed, the number of distinct eigenvalues of $\mQ^X_I$ and $\mQ^Y_I$ are $2L$ rather than $L+1$ for the regular $U(1)$ symmetry.
In particular, the distinct eigenvalues of the operator $\mQ^X_I(q)$ are $\{-L+2(n - mq^L) + (1 - q^L)\}$ for integer $n,m$ from ranges $0 \leq n \leq L-1$ and $0 \leq m \leq 1$ [and similar for $\mQ^Y_I(q)$ with $q \rightarrow q^{-1}$], which can be computed using their expressions as quadratic operators in terms of Majorana fermions, as discussed in App.~\ref{subsec:U1symmetries}.
As we show in App.~\ref{app:hiddenU1}, these can be compactly expressed as Matrix Product Operators (MPOs), which also allows an easier proof that they are in the commutant $\mC_{U(1), I}$. 
In the presence of the SZMs $\Psi^\ell(q)$ and $\Psi^r(q)$ and the global $\mathbb{Z}_2$ symmetry, $\mQ^X_I$ and $\mQ^Y_I$ are not independent though; specifically, we find the relation
\begin{equation}
    q \mQ^X_I - q^{-1} \mQ^Y_I = \frac{\eta Q_Z \Psi^\ell(q) \Psi^r(q)}{q^{L-1}}.
\label{eq:U1relations}
\end{equation}
Hence it is sufficient to include a single quasilocal $U(1)$ generator in the commutant $\mC_{U(1), I}$ of $\mA_{U(1), I}$, which we conjecture to be
\begin{equation}
    \mC_{U(1), I}  = \lgen Q_Z, \Psi^\ell(q), \Psi^r(q), \mQ_I(q) \rgen,
\label{eq:quasiloccomm}
\end{equation}
where $\mQ_I$ can be either $\mQ^X_I$ or $\mQ^Y_I$. 
We also note that unlike the regular $U(1)$ symmetry discussed in Sec.~\ref{subsubsec:regularU1}, $\mQ_I(q)$ cannot be a symmetry of a locally generated bond algebra on its own -- any local Hamiltonian with that symmetry will necessarily have additional symmetries, such as the $\mathbb{Z}_2$ and SZMs.
We discuss this aspect more in App.~\ref{app:hiddenU1}.
Note that with this choice of $(\kappa, \gamma)$, Eq.~(\ref{eq:FMcondition}) is satisfied for any $|q| \neq 1$.
Hence the SZMs are localized on different sides of the chain, implying that the ground state of a non-interacting Hamiltonian such as Eq.~(\ref{eq:XYHamil}) with these parameters is in the ferromagnetic phase~\cite{franchini2017integrable}. 
Interestingly, such a Hamiltonian naturally appears in the study of entanglement dynamics of some ensembles of random circuits~\cite{vardhan_entanglement_2024, suzuki2024more, znidaric2008exact}, where late-time entanglement physics demands that the system is in its ferromagnetic phase.
However, the quasi-local $U(1)$ symmetry was not pointed out in such contexts, and its implication for ground state physics is not immediately clear. 
There are potentially hydrodynamical effects of such a $U(1)$ symmetry that could be relevant for non-integrable models built from the bond algebra $\mA_{U(1), I}$ of Eq.~(\ref{eq:U1bondI}), and we will discuss them in Sec.~\ref{subsec:nonintSZM}.
\subsubsection{Staggered quasi-local \texorpdfstring{$U(1)$}{}}\label{subsubsec:staggeredquasilocalU1}
We also observe a hidden $U(1)$ symmetry when $q(\kappa + \gamma) = - q^{-1}(\kappa - \gamma)$, which we can w.l.o.g. set to $1$ [or equivalently $(\kappa, \gamma) = (-\frac{q - q^{-1}}{2}, \frac{q + q^{-1}}{2})$]. 
Here we obtain the bond algebra
\begin{equation}
    \mA_{U(1),II} = \lgen \{q^{-1} X_j X_{j+1} - q Y_j Y_{j+1} + (Z_j - Z_{j+1}) \}\rgen.
\label{eq:AU1II}
\end{equation}
Since this is also a special case of the XY bond algebra of Sec.~\ref{subsec:XYSZM}, we again have the conserved quantities $\Psi^\ell(q)$ of Eq.~(\ref{eq:ISZMdefns}), and in addition we have the conserved quantity $\Psi^r(s = -q)$ of Eq.~(\ref{eq:PsiR}).
Following the same steps as in Sec.~\ref{subsubsec:regularquasilocalU1}, we find the conserved quantities
\begin{align}
    \mathcal{Q}^X_{II}(q) &= \sum_{j=1}^L{(-1)^j Z_j} - \eta \sum_{1 \leq j_1 < j_2 \leq L}{X_{j_1} \left(\prod_{k = j_1+1}^{j_2 - 1} q Z_{k} \right) X_{j_2}}\nn \\
    \mathcal{Q}^Y_{II}(q) &= \sum_{j=1}^L{(-1)^j Z_j} + \eta \sum_{1 \leq j_1 < j_2 \leq L}{Y_{j_1} \left(\prod_{k = j_1+1}^{j_2 - 1} q^{-1} Z_{k} \right) Y_{j_2}}
\label{eq:QXsymopII}
\end{align}
Note that in the $q = 1$ limit, this reduces to the staggered $U(1)$ symmetry evident from the generators of $\mA_{U(1),II}$ in Eq.~(\ref{eq:AU1II}).
In fact, this bond algebra $\mA_{U(1), II}$ can be exactly mapped to the generators of the bond algebra $\mA_{U(1),I}$ of Eq.~(\ref{eq:U1bondI}) with the on-site unitary transformation of Eq.~(\ref{eq:stagunitary2}).
This means that all the algebraic relations and the physics of systems with this staggered quasi-local symmetry can be exactly mapped to that of the regular quasi-local $U(1)$ symmetry discussed above.
\section{Non-Integrable models with SZMs}\label{subsec:nonintSZM}
\begin{figure}
\includegraphics[scale=0.95]{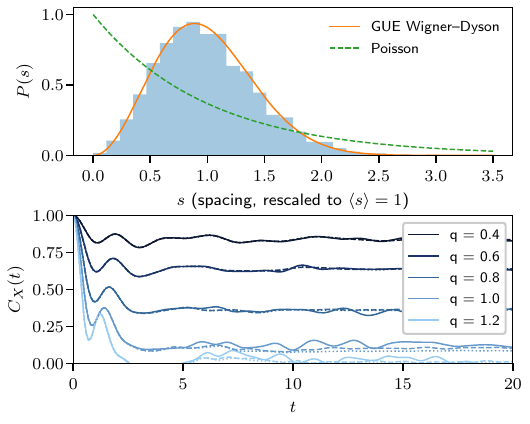}
\caption{
\textbf{Non-integrable models with the Ising SZM.} 
Data shown for the non-integrable model of Eq.~(\ref{eq:HintIsingSZM}) that possesses the Ising SZM $\Psi^\ell(q)$ of Eq.~(\ref{eq:ISZMdefns}) as an exact conserved quantity.
(Top) Statistics of energy level differences $\{s_n = E_{n+1} - E_n\}$ (normalized such that the mean spacing $\langle s \rangle = 1$) in the middle-half of the spectrum in the $Q_Z = +1$ sector for parameters $(q, c, r) = (0.55, 0.53, 0.23)$ and system size $L = 14$.
It clearly exhibits Wigner-Dyson statistics from the GUE ensemble, characteristic of non-integrable models that break time-reversal symmetry, in contrast to Poisson level statistics exhibited by integrable models.
(Bottom) Autocorrelation function $C_{X}(t) = 2^{-L} \text{Tr}[X_1(t) X_1(0)]$ of the operator $X_1$ on the leftmost site for parameters $(c, r) = (0.53, 0.23)$ and different values of $q$ from $0.4$ to $1.2$ and system sizes $L = 8, 10, 12$ (different linestyles---continuous, dashed, dotted). 
The autocorrelation functions show saturation to a finite value of $\sim 1 - |q|^2$ for $q < 1$, indicating an SZM localized at the left edge, while for $q > 1$ the autocorrelations quickly decay to $0$.}
\label{fig:nonintegrable}
\end{figure}
\subsection{General Construction}
With the commutant understanding of SZMs, we now show that it is straightforward to construct non-integrable models with exact SZMs.
In all the cases discussed in this section, the SZMs have bond algebras generated by terms that are fermionic bilinears after a Jordan-Wigner transformation, hence linear superpositions of these generators are always integrable.
Such models have been the focus of most previous studies of SZMs~\cite{fendley2012parafermionic, fendley2016strong, katsura2015exact, kawabata2017exact, kemp2017long, maceira2018infinite, wouters2018exact, vasiloiu2018enhancing, vernier2024strong}.
However, the bond algebras also contain \textit{products} of the generators, which allows us to construct non-integrable models.
While the product of two non-commuting Hermitian operators $A$ and $B$ is not Hermitian in general, their commutator $i [A, B]$ and their anti-commutator $\{A, B\}$ are Hermitian. 
Note that the commutators of quadratic fermionic terms are also quadratic fermionic terms, hence to obtain interacting terms, we need to consider the anti-commutators of the generators and their linear combinations. 
Given a set of free-fermion generators $\{T_{j,j+1}\}$, we can always build an interacting local Hamiltonian as
\begin{equation}
    H_{\textrm{int}} \defn \sum\nolimits_{j}{T_{j,j+1}} + \frac{c}{2} \sum\nolimits_{j}{\{T_{j,j+1}, T_{j+1, j+2}\}}, 
\label{eq:hamint}
\end{equation}
where $c$ is an arbitrary constant.
$H_{\mathrm{int}}$ would have at least the same symmetries as the bond algebra generated by $\{T_{j,j+1}\}$, including the SZMs. 
For a simple example, we can choose $T_{j,j+1}$ to be the generators of $\mA_{\iszm}$ of Eq.~(\ref{eq:Isingbond}), use Eq.~(\ref{eq:hamint}), and also add the terms $\{X_j Y_{j+1}\}$ to break the trivial $X_L$ symmetry as in Eq.~(\ref{eq:IsingXLbroken}). 
We then obtain an interacting Hamiltonian
\begin{gather}
    H_{\textrm{I-int}} =  -\sum_{j=1}^{L-1}{(X_j X_{j+1} + q Z_j + r X_j Y_{j+1})} \nn \\
    - c \sum_{j = 1}^{L-2}{(X_j X_{j+2} + q^2 Z_j Z_{j+1} + q Z_j X_{j+1} X_{j+2}),}
\label{eq:HintIsingSZM}
\end{gather}
which by construction commutes with the operators in $\widetilde{\mC}_{\iszm}$ of Eq.~(\ref{eq:IsingXLbroken}).
Hence this model has a global $\mathbb{Z}_2$ symmetry generated by $Q_Z$ and an SZM $\Psi^\ell(q)$, which anticommute between themselves, leading to a double degeneracy of the spectrum of $H_{\text{I-int}}$. 
As shown in Fig.~\ref{fig:nonintegrable}a, it is easy to numerically verify that the eigenvalues of $H_{\textrm{I-int}}$ within a given $\mathbb{Z}_2$ symmetry sector exhibit level repulsion and Wigner-Dyson statistics, a hallmark of non-integrable models.
Hence this is a non-integrable model with an exact SZM $\Psi^{\ell}(q)$ of Eq.~(\ref{eq:ISZMdefns}).
For $|q| < 1$, this SZM is localized on the left end of the chain, as reflected in the autocorrelation function of the leftmost spin, shown in Fig.~\ref{fig:nonintegrable}b, and also discussed in Sec.~\ref{subsubsec:edgecorr}.
It can be made an approximate SZM with the addition of any local terms on the right end of the chain, similar to the discussion in Sec.~\ref{subsubsec:approxSZM}.
For example, we can add terms $q Z_L$ and $q^2 Z_{L-1} Z_L$ to Eq.~(\ref{eq:HintIsingSZM}) to ensure that the couplings are ``uniform" throughout the chain, i.e.,
\begin{equation}
    \widetilde{H}_{\text{I-int}} = H_{\text{I-int}} - qZ_L - q^2 Z_{L-1}Z_L.
\label{eq:HintapproxIsing}
\end{equation}
This $\widetilde{H}_{\text{I-int}}$ would still have $\Psi^\ell(q)$ of Eq.~(\ref{eq:ISZMdefns}) as an approximate SZM.
Such a system is a clear example of a uniform Ising model with non-integrable interactions that possesses an approximate SZM.
All eigenstates of $H_{\textrm{I-int}}$ have an exact two-fold degeneracy, whereas the eigenstates of $\widetilde{H}_{\textrm{I-int}}$ have an approximate two-fold degeneracy (in the sense of the discussion in Sec.~\ref{subsubsec:approxSZM}).
However, the two-fold degeneracies need not translate to ferromagnetic order in the eigenstates, even in the presence of disorder, due to questions on the stability of Many-Body Localization (MBL)~\cite{huse2013localization, suntajs2020quantum, sierant2025manybody}.
For the ground state, since we know that the $c = 0$ (non-interacting limit) is in the ferromagnetic phase (as discussed in Sec.~\ref{subsec:phaseconnection}), we are guaranteed that weak $\mathbb{Z}_2$-symmetric perturbations such as the second line in Eq.~(\ref{eq:HintIsingSZM}) should still preserve the ferromagnetic order.
However, we are generally not able to rule out the possibility of phase transitions into a paramagnetic or some other phase upon increasing the magnitude of $c$, particularly since $c$ can be either positive or negative.
It would be interesting to more carefully explore the phase diagrams of such models in the future.
Similar constructions can work for the bond algebras discussed in the context of the XY SZM in Sec.~\ref{sec:XYSZM}.
The exactness of the SZMs [and the quasi-local $U(1)$ symmetry in the case of Sec.~\ref{subsec:quasilocalU1I}] in those cases too require rather fine-tuned boundary terms, as in the Hamiltonian $H_{\textrm{I-int}}$ of Eq.~(\ref{eq:HintIsingSZM}).
However, in the case of the XY SZM, it is also hard to construct a model with uniform couplings throughout the chain that have an approximate SZMs, analogous to $\widetilde{H}_{\textrm{I-int}}$ of Eq.~(\ref{eq:HintapproxIsing}). 
This is because in the XY case, terms need to be added on both ends of the chain to make the couplings uniform (see discussion in Sec.~\ref{subsubsec:XYconnections}), which also kills the approximateness of the SZMs.
Finally, we note that these constructions of large families of Hamiltonians that share certain SZMs are closely related to reverse-engineering of Hamiltonians from zero-modes demonstrated in \cite{chertkov2020engineering}, where methods for constructing non-interacting and interacting (hence presumably non-integrable) models with a desired ``topological'' mode were explored and demonstrated.
\subsection{Implications for Dynamics}\label{subsec:dynamics}
\subsubsection{Correlation Functions and Mazur Bounds}
Given that the SZMs are conserved quantities that are exponentially localized, they should have some imprint on the dynamics of correlation functions.  
For example, the Ising SZM is known to lead to the saturation or pre-thermal plateaus in the autocorrelation function of the edge operator $X_1$~\cite{kemp2017long, maceira2018infinite, yates2020dynamics, vasiloiu2019strong}.
Such phenomena are easy to understand using the so-called Mazur bound~\cite{mazurbound1969, dhar2020revisiting, moudgalya2021hilbert}, which says that the time-averaged autocorrelation function of some operator $\widehat{O}$ can be bounded from below as:
\begin{equation}
    \lim_{T \rightarrow \infty}\int_{0}^{T}{\obraket{\widehat{O}(t)}{\widehat{O}(0)}} \geq \sum_{\alpha}{\frac{|\obraket{\widehat{O}}{Q_\alpha}|^2}{\obraket{Q_\alpha}{Q_\alpha}}},
\label{eq:Mazurbound}
\end{equation}
where we have defined $\obraket{\widehat{A}}{\widehat{B}} \defn \Tr[\widehat{A}^\dagger \widehat{B}]/D$ where $D$ is the Hilbert space dimension, and $Q_\alpha$'s are the orthogonalized conserved quantities of the system.
While usually only quasi-local symmetries are included in the sum in Eq.~(\ref{eq:Mazurbound}), there are many systems where the inclusion of many kinds of non-local symmetries can be fruitful to obtain a bound, e.g., in the case of Hilbert space fragmentation~\cite{rakovszky2020statistical, moudgalya2021hilbert}.
These dynamical signatures of conserved quantities such as SZMs can be much more striking for the non-integrable models than for their integrable counterparts, since correlation functions in the latter might be affected by other special features such as the non-interacting property, rather than just the SZM. 
For example, in non-integrable systems, we expect that the Mazur bound of Eq.~(\ref{eq:Mazurbound}) saturates at late times provided all known conserved quantities are taken into account.  
This can be made more precise by using simplified models for the dynamics of a non-integrable system, such as random unitary circuits~\cite{fisher2023random}.
In particular, the correlation functions can be \textit{averaged} over the randomness in the circuit to enable more precise calculations of the dynamics, with the assumption that the average is representative of the typical behavior (which can also be checked in principle with more sophisticated computations). 

\begin{figure}
\includegraphics[scale=0.95]{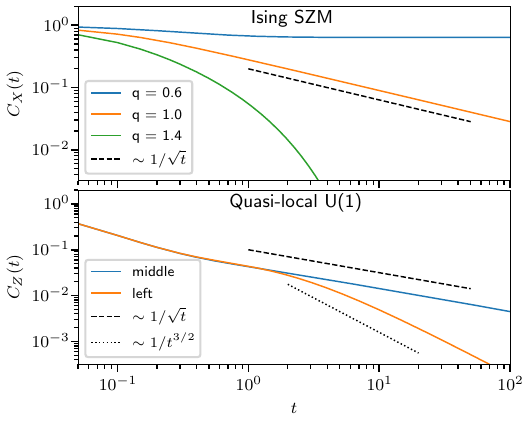}
\caption{\textbf{Dynamics of autocorrelation functions in Brownian circuits with SZMs.}
Computations were done using Time-Evolving Block Decimation (TEBD) on Eq.~(\ref{eq:avgcorrfunc}) setting $g = 1$ and for system size $L = 100$ and maximum bond dimension $\chi = 50$. 
(Top) 
$C_{X_1}(t) = 2^{-L} \text{Tr}[X_1(t) X_1(0)]$ of the $X_1$ operator on the left edge spin in a Brownian circuit constructed from the generators of the Ising bond algebra $\mA_{\iszm}$ of Eq.~(\ref{eq:Isingbond}) for three values of $q$.
This shows saturation to a constant ($q < 1$), diffusive power-law decay ($q = 1$), and exponential decay ($q > 1$), which respectively correspond to the Ising SZM being localized on the left edge of the chain, delocalized and possessing hydrodynamic modes, or localized on the right edge of the chain.
(Bottom) 
$C_{Z_j}(t) = 2^{-L} \text{Tr}[Z_j(t) Z_j(0)]$ for the operator $Z_j$ on the left (site $j = 4$) or middle (site $j = 50$) of the chain in a Brownian circuit constructed from the generators of the bond algebra $\mA_{U(1),I}$ of Eq.~(\ref{eq:U1bondI}) as well as boundary terms $\{Z_1\}$ and $\{Z_L\}$ that break the quasi-local $U(1)$ symmetry. 
$Z_{j=50}$ in the middle exhibits hydrodynamics similar to a system with regular $U(1)$ conservation with a clear diffusive power-law decay up to large timescales $\sim O(L^2)$. 
$Z_{j=4}$ close to the boundary of the system senses the breaking of the quasilocal symmetry on the boundary at finite timescales $\sim O(j^2)$, and shows a power-law decay $\sim t^{-\frac{3}{2}}$ corresponding to the physics of a diffusing particle with an absorbing boundary, as discussed in Sec.~\ref{subsubsec:quasilocalhydro}.}
\label{fig:brownianSZM}
\end{figure}
\subsubsection{Brownian Circuit Formalism}\label{subsubsec:Brownian}
For mimicking the dynamics of a local non-integrable system constructed from a bond algebra $\mA = \lgen \{T_{j,j+1}\}\rgen$, random evolutions known as Brownian circuits~\cite{lashkari2013towards, bauer2017stochastic, sunderhauf2019quantum, xu2019locality} have been shown to be effective~\cite{moudgalya2023symmetries}.
Given a set of generators $\{T_{j,j+1}\}$, these systems consist of time-evolution by global unitaries composed of local time-dependent unitaries of the form
\begin{equation}
    U_{\delta t} = e^{i \sum_{j}{J^{(t)}_j T_{j,j+1}}\delta t},\;\;\;\overline{J^{(t)}_j J^{(t')}_k} = \frac{2g}{\delta t} \delta_{j,k}\delta_{t,t'},
\label{eq:Brownianunitaries}
\end{equation}
where the $J^{(t)}_{j}$'s are Gaussian random variables with the shown correlation, and we take the $\delta t \rightarrow 0$ limit.
The averaged behavior (over the randomness in the circuits, denoted $\overline{\cdots}$) of the autocorrelation function of any operator $\widehat{O}$ can then be shown to be~\cite{lashkari2013towards, bauer2017stochastic, sunderhauf2019quantum, ogunnaike2023unifying, moudgalya2023symmetries}
\begin{equation}
    \overline{\obraket{\widehat{O}(t)}{\widehat{O}(0)}} = \obra{\widehat{O}(0)}e^{-g\hmP t}\oket{\widehat{O}(0)}, 
\label{eq:avgcorrfunc}
\end{equation}
where $\hmP$ is a superoperator defined as
\begin{equation}
    \hmP = \sum\nolimits_{j}{\hmL_{j,j+1}^\dagger \hmL_{j,j+1}},\;\;\;\hmL_{j,j+1} = T_{j,j+1} \otimes \mathds{1} - \mathds{1} \otimes T_{j,j+1}^T.
\label{eq:superhamiltonian}
\end{equation}
Here we have also used Choi-Jamiolkowski
operator-to-state mapping and have written the final expression in the doubled-space notation.
It is easy to show that the ground states of $\hmP$ are simply the symmetry operators in the commutant $\mC$ of $\mA$~\cite{moudgalya2023symmetries}, so choosing an orthogonal basis $\oket{Q_\alpha}$ for the commutant, the $t \rightarrow \infty$ limit of Eq.~(\ref{eq:avgcorrfunc}) reads
\begin{equation}
    \lim_{t \rightarrow \infty}{\overline{\obraket{\widehat{O}(t)}{\widehat{O}(0)}}} = \sum_{\alpha}{\frac{|\obraket{\widehat{O}}{Q_\alpha}|^2}{\obraket{Q_\alpha}{Q_\alpha}}},
\label{eq:inftime}
\end{equation}
which is precisely the R.H.S.\ of Eq.~(\ref{eq:Mazurbound}).
This shows that autocorrelation functions in these random circuits, when averaged over the randomness, saturate the corresponding Mazur bounds. 
From Eq.~(\ref{eq:avgcorrfunc}), it is also clear that the approach to this stationary value is governed by the low-energy excitations of $\hmP$, which are the slow-modes corresponding to the symmetries in $\mC$~\cite{ogunnaike2023unifying, moudgalya2023symmetries}.
These can sometimes be variationally deduced by studying the structure of the symmetry operators $\oket{Q_\alpha}$, which are the ground states of $\hmP$, and we will discuss some such examples below.
When all the $\{T_{j,j+1}\}$ are quadratic fermion operators after a Jordan-Wigner transform (which holds for all the bond algebras studied in this work), it turns out that $\hmL_{j,j+1}$ of Eq.~(\ref{eq:superhamiltonian}) has a $U(1)$ superoperator symmetry~\cite{lastres2024nonuniversality} (which also can be referred to as a \textit{weak} $U(1)$ symmetry when $\hmP$ is interpreted as a Lindbladian~\cite{buca2012note}).
Such a superoperator symmetry also carries forward to the super-Hamiltonian $\hmP$, which can be used to block-diagonalize $\hmP$ into subspaces invariant under this superoperator $U(1)$, sometimes allowing for a simpler solution for the commutants and the hydrodynamic modes. 
We discuss this in more detail in App.~\ref{subsec:freefermsuper} and also demonstrate some examples.
However, note that the Brownian circuits constructed out of free-fermion terms are expected to exhibit averaged behavior that is similar to the interacting Brownian circuits that have the same commutant, when it comes to simple correlation functions such as Eq.~(\ref{eq:avgcorrfunc}).
This is related to the fact that free-fermion circuits (also known as matchgate circuits) with any symmetry form unitary 1-designs with respect to that symmetry~\cite{lastres2024nonuniversality}.
\subsubsection{Edge Autocorrelations with the Ising SZM}\label{subsubsec:edgecorr}
We now study the dynamics of correlation functions expected in non-integrable systems with the Ising SZM.
We start with a Brownian circuit built from the generators of the bond algebra $\mA_{\iszm}$ of Eq.~(\ref{eq:Isingbond}), and focus on the autocorrelation function of the edge operator $X_1$.
Its randomness-averaged infinite-time value can be computed using Eq.~(\ref{eq:inftime}), choosing the $\oket{Q_\alpha}$'s from Eq.~(\ref{eq:ISZMcomm}), which can be verified to be an orthogonal basis (although not normalized) for the commutant. 
Noting that the operator $X_1$ is orthogonal to all of them except $\Psi^\ell(q)$, we obtain that
\begin{equation}
    \lim_{t \rightarrow \infty}\overline{\obraket{X_1(t)}{X_1(0)}} = \frac{|\obraket{X_1}{\Psi^\ell(q)}|^2}{\obraket{\Psi^\ell(q)}{\Psi^\ell(q)}} = \frac{1 - |q|^2}{1 - |q|^{2L}},
\label{eq:Psicontr}
\end{equation}
with the $|q| = 1$ case being the formal $q \rightarrow 1$ limit of the expression, which is $1/L$.  
Intuitively, in the $L \rightarrow \infty$ limit, this takes a non-zero value of $1 - |q|^2$ only when $|q| < 1$, which is when the SZM is localized on the left side of the chain, and this is what we indeed find in direct computations using the corresponding super-Hamiltonian as shown in the top panel of Fig.~\ref{fig:brownianSZM}.
This is also consistent with the dynamics observed in the non-integrable Hamiltonian with the Ising SZM, shown in the bottom panel of Fig.~\ref{fig:nonintegrable}.
The approach to this infinite-time value can be deduced from the gap of the corresponding super-Hamiltonian $\hmP_{\iszm}$ of Eq.~(\ref{eq:superhamiltonian}) constructed from the generators of $\mA_{\iszm}$ in Eq.~(\ref{eq:Isingbond}), and in App.~\ref{app:szmhydro} we illustrate two methods to do this.
First, we show that a natural set of trial states that are orthogonal to its ground states [which are the operators in Eq.~(\ref{eq:ISZMcomm})] are of the form $\oket{\Psi^\ell(z_k)}$ with $z_k = q^{-1} e^{ik}$, where we have implicitly generalized Eq.~(\ref{eq:ISZMdefns}) to complex arguments.
We find that this state has a trial energy of $E_k = 4|q - z_k|^2$ under $\hmP_{\iszm}$, which suggests that it is gapped when $|q| \neq 1$, consistent with expectations that the SZMs do not have slow-modes by virtue of being a discrete symmetry. 
Second, we exploit the $U(1)$ superoperator symmetry of $\hmP_{\iszm}$, inherited from the fact that $\mA_{\iszm}$ is generated by free-fermion terms, and use it to solve for some of the exact low-lying excitations of $\hmP_{\iszm}$.
However, this second method works only for this specific choice of generator for $\mA_{\iszm}$, whereas the variational method would work even for a more complicated set of generators that could be interacting in the fermion language.
Interestingly, the trial state calculation also means that when $q = 1$, $\oket{\Psi^{\ell}(z_k = e^{ik})}$ for $k = 2\pi/L$ is a gapless mode with energy $\sim k^2$, hence $\hmP_{\iszm}$ is gapless. 
This can also be verified using the exact computations for the partial spectrum of $\hmP_{\iszm}$ using its symmetries, as explained in App.~\ref{subsec:freefermsuper}.
Analytical computations with these Brownian circuits using these low-lying hydrodynamic modes, Eq.~(\ref{eq:avgcorrfunc}), and the fact that $X_1$ has an important overlap on these hydrodynamic modes, suggest that schematically, $\overline{\obraket{X_1(t)}{X_1(0)}} \sim \int dk\ e^{-k^2 t} \sim t^{-\frac{1}{2}}$ for large $t$.
[In fact, fermionic operators in the bulk, $(\prod_{j'=1}^{j-1} Z_{j'}) X_j$ ($= \chi_{2j-1}$, see App.~\ref{app:jordanwigner}), also show similar power-law decay of the autocorrelations in the $q=1$ case.]
This can also be derived more rigorously using results in App.~\ref{app:szmhydro}, and along the lines of computations discussed in much more detail in Ref.~\cite{moudgalya2023symmetries} for autocorrelation functions of charge densities in systems with a local conserved charge.
This can also be verified numerically as shown in Fig.~\ref{fig:nonintegrable}b.
In the future it would be interesting to check for signs of this in non-integrable models such as $H_{\text{I-int}}$ of Eq.~(\ref{eq:HintIsingSZM}).
We should emphasize that these behaviors do not necessarily hold in the translationally invariant (``clean'') free-fermion Hamiltonian model such as for Eq.~(\ref{eq:TFIMhamil}) due to the presence of numerous other conserved quantities.
Nevertheless, we numerically find the same qualitative behaviors in the clean model that hold on ``average" or after time-averaging the correlation function over small windows, in particular an exponential decay when $|q| \neq 1$, and a power-law decay when $q = 1$ although with a different power for the edge autocorrelations.
It would be interesting to further explore the role of $\Psi^\ell(q)$ in such cases. 
A similar analysis can be performed to reveal the approximate SZM that survives under the addition of perturbations on the right end of the chain when $|q| < 1$, such as in the presence of the rightmost field $Z_L$ discussed in Sec.~\ref{subsubsec:approxSZM}, see also Eq.~(\ref{eq:TFIMhamil}). 
A simple toy model that captures the dynamics of correlation functions in the presence of such a field is a Brownian circuit that is generated from the generators of Eq.~(\ref{eq:Isingbond}) and the $Z_L$ term on the right edge. 
The effective super-Hamiltonian in this case is simply 
\begin{equation}
    \hmP_{\tot} = \hmP_{\iszm} + \hmP_{\rm bdy},
\label{eq:bdypert}
\end{equation}
where $\hmP_{\iszm}$ [resp. $\hmP_{\rm bdy}$] is of the form of Eq.~(\ref{eq:superhamiltonian}) with $T_{j,j+1}$ chosen from Eq.~(\ref{eq:Isingbond}) [resp. with $Z_L$ instead of $T_{j,j+1}$]. 
As discussed earlier, the ground states of $\hmP_{\iszm}$ are simply the operators in the commutant $\mC_{\iszm}$, and the effect of the addition of $\hmP_{\rm bdy}$ can be computed using degenerate perturbation theory on this ground state space, similar to recent computations done to study the dynamics of correlation functions in symmetric systems in the presence of symmetry-breaking impurities~\cite{li2025dynamics}. 
For $|q|< 1$, the addition of $\hmP_{\rm bdy}$ leads to only an exponentially small lifting for some of the ground states of $\hmP_{\iszm}$, with $\oket{\mathds{1}}$ and $\oket{Q_Z}$ still being the zero-energy ground states of $\hmP_{\tot}$ [since they exactly commute with $Z_L$], and with $\oket{\Psi^\ell(q)}$ and $\oket{Q_Z \Psi^\ell(q)}$ acquiring an energy $\sim q^L$ [since they approximately commute with $Z_L$ as seen from Eq.~(\ref{eq:approxSZMcondition})], while all other operators are gapped. 
Using Eq.~(\ref{eq:avgcorrfunc}), this exponentially small splitting implies that there is an exponentially long pre-thermal plateau in the averaged autocorrelation function of $X_1$, which is governed by the contribution $\Psi^\ell(q)$ shown in Eq.~(\ref{eq:Psicontr}).   
On timescales $t \gg O(1/q^L)$, the autocorrelation function is expected to decay to zero as $\sim e^{-c q^L t}$, similar to the case with local integrals of motion in the examples of Hilbert space fragmentation~\cite{sala2020fragmentation,khemani2020localization,  moudgalya2019thermalization,rakovszky2020statistical, moudgalya2021hilbert} discussed in \cite{li2025dynamics}.
Similar stability analysis can be performed for other examples of SZMs discussed in Sec.~\ref{subsec:XYSZM}, which we will not describe here in detail.
\subsubsection{Hydrodynamics with the Quasi-local \texorpdfstring{$U(1)$}{} Symmetry}\label{subsubsec:quasilocalhydro}
We now discuss the hydrodynamics of non-integrable systems that exhibit the quasi-local $U(1)$ symmetries discussed in Sec.~\ref{subsec:quasilocalU1I}.
Since the conserved quantities in Eq.~(\ref{eq:QXsymop}) are quasilocal operators that resemble the regular $U(1)$ symmetry, we expect the hydrodynamics to be similar to those with regular $U(1)$ such as those discussed in Sec.~\ref{subsec:conventionalU1}, which is very well-studied. 
We focus on the autocorrelation function $\obraket{Z_j(t)}{Z_j}$ for $Z_j$ in the bulk of the system, which, as we will demonstrate, exhibits interesting hydrodynamic behavior. 
Following Eq.~(\ref{eq:avgcorrfunc}), we will consider a Brownian circuit built from the generators of $\mA_{U(1),I}$ of Eq.~(\ref{eq:U1bondI}), and we refer to the corresponding super-Hamiltonian of Eq.~(\ref{eq:superhamiltonian}) as $\hmP_{U(1),I}$. 
The ground states of $\hmP_{U(1),I}$ are the $4L$ elements in the commutant $\mC_{U(1),I}$ of Eq.~(\ref{eq:quasiloccomm}).
We now describe the structure of the low-energy excitations of $\hmP_{U(1),I}$ for large system sizes.
This is expected to be similar to systems with regular $U(1)$ conservation, so we will be schematic below, and the reader can refer to \cite{moudgalya2023symmetries} for more precise details.
One of the ground states of $\hmP_{U(1),I}$ is the quasilocal conserved quantity itself, which can be expressed as 
\begin{equation}
    \mQ_{U(1),I} = \sum\nolimits_j{\widetilde{Z}_j} + \cdots_{\rm bdy},
\label{eq:GSstruct}
\end{equation}
where $\widetilde{Z}_j$ is an operator that is exponentially localized around site $j$ and $\cdots_{\rm bdy}$ are corrections near the boundary of the system.
From Eq.~(\ref{eq:QXsymop}), we can choose $\widetilde{Z}_j$ to be some quasilocal term that is symmetric about site $j$, e.g., for $|q| < 1$, $\mQ^X_I$ is a quasilocal conserved quantity of the form of Eq.~(\ref{eq:GSstruct}) with
\begin{equation}
    \widetilde{Z}_j = Z_j - \frac{\eta}{2} (X_{j-1} X_j + X_j X_{j+1}) - \eta q X_{j-1} Z_j X_{j+1} - \cdots,
\label{eq:quasilocaldensity}
\end{equation}
whereas when $|q| > 1$, $\mQ^Y_{I}$ is the quasilocal conserved quantity with a similar expression. 
Given the ground state in Eq.~(\ref{eq:GSstruct}) of the super-Hamiltonian $\hmP_{U(1),I}$, we expect that plane-waves of the form $\widetilde{Z}_k = \sum_{j}{e^{ikj}\widetilde{Z}_j}$ are gapless excitations of $\hmP_{U(1),I}$ with energy $\sim k^2$.
All this is exactly analogous to a regular $U(1)$ symmetry~\cite{moudgalya2023symmetries}, where $\widetilde{Z}_j = Z_j$.  
We then observe that the operator $Z_j$, with appropriate normalizations, has an $O(1)$ overlap with the quasiparticle operator $\widetilde{Z}_j$, which also means that the total overlap of $Z_j$ on all the plane-wave states $\{\widetilde{Z}_k\}$ is $O(1)$. 
Since these plane-waves are also the low-energy modes of $\hmP_{U(1),I}$, we expect that the expression for $\overline{\obraket{Z_j(t)}{Z_j}}$ using Eq.~(\ref{eq:avgcorrfunc}) will be dominated by the contribution from these plane-waves $\{\widetilde{Z}_k\}$.
Given their $\sim k^2$ dispersion, this gives $\overline{\obraket{Z_j(t)}{Z_j(0)}} \sim \int dk\ e^{-k^2 t} \sim t^{-\frac{1}{2}}$, similar to the case of a regular $U(1)$ symmetry~\cite{ogunnaike2023unifying, moudgalya2023symmetries}. 
This behavior can be explicitly verified for Brownian circuit evolution of a $Z_j$ operator deep in the bulk of the system, as shown in Fig.~\ref{fig:nonintegrable}c.
We now discuss the addition of boundary perturbations that break the quasilocal $U(1)$ symmetry, which are naturally needed if we wish to construct models with such a $\mQ_{I}^{X}$ or $\mQ_{I}^{Y}$ quasi-local conserved quantity \textit{and} want to make the couplings and fields uniform throughout the chain. 
We can then explore the effect of such uniformization of the chain on the correlation function $\obraket{Z_j(t)}{Z_j(0)}$. 
An analogous situation was studied in \cite{li2025dynamics}, which explored the behavior of the autocorrelation function $\obraket{Z_j(t)}{Z_j(0)}$ in a one-dimensional system with a regular $U(1)$ conserved charge $\sum_j{Z_j}$ that is broken by a term such as $X_1$ at the edge of the chain.
Using Brownian circuits with appropriate generators, it showed that the dynamics of such a correlation function can be mapped onto the probability distribution of a particle initialized at position $j$ diffusing in one spatial dimension with \textit{absorbing} boundary conditions. 
For $j \ll L$, this gives rise to the decay of autocorrelation functions of the form $\overline{\obraket{Z_j(t)}{Z_j(0)}} \sim t^{-\frac{3}{2}}$ for timescales $O(j^2) \ll t \ll O(L^2)$.
Following the same set of arguments, we expect the same behavior for the autocorrelation function in the presence of the quasi-local conserved quantity that is broken at the edges.
Indeed this can be verified in Brownian circuits as shown in the bottom panel of Fig.~\ref{fig:brownianSZM}, where we use Brownian gates generated by the generators of $\mA_{U(1),I}$ in Eq.~(\ref{eq:U1bondI}), as well as those generated by boundary terms $\{Z_1, Z_L\}$ that break the  quasi-local $U(1)$  symmetry.
We also expect all this physics to persist in non-integrable models built from the generators of $\mA_{U(1), I}$ of Eq.~(\ref{eq:U1bondI}) and using the prescription such as Eq.~(\ref{eq:hamint}).
Admittedly, this is harder to see using simple exact diagonalization due to the limited system sizes accessible, but it would be interesting in the future to perform a thorough numerical study using more advanced technique to check if the physics described above is visible.
\subsection{Strong Zero Modes without Degeneracies}
Finally, as a quick comment, we point out that breaking integrability while preserving the SZMs also allows us to construct SZMs without the accompanying two-fold degeneracy. 
All the commutants we studied in Secs.~\ref{sec:Ising} and \ref{sec:XYSZM} contain SZMs along with a $\mathbb{Z}_2$ symmetry operator $Q_Z$ that anti-commutes with the SZMs, which leads to exact double degeneracies in the spectra of the operators constructed from the bond algebra.
The existence of the $Q_Z$ symmetry there simply follows from the fact that we considered algebras that are generated by terms that, after a Jordan-Wigner transformation, map onto quadratic free-fermion terms, whose $\mathbb{Z}_2$ parity symmetry is $Q_Z$ in the spin language.
Going beyond such free-fermion generators allows us to preserve the SZMs while breaking the underlying $\mathbb{Z}_2$ symmetry of the system.
For example, if we add $X_{L-1} Z_L$ at the right edge of the chain to the bond algebra $\mA_{\iszm}$ of Eq.~(\ref{eq:Isingbond}), $Q_Z$ is no longer in the resulting commutant, which is only generated by $\Psi^\ell(q)$.
Hence Hamiltonians constructed out of the generators of such bond algebras do not have double degeneracies in their spectra, but yet exhibit the dynamical signatures of SZMs discussed in Sec.~\ref{subsec:dynamics}.
Such commutants closely resemble those that appear in systems with an \textit{exponential symmetry}~\cite{sala2022dynamics, lehmann2023fragmentation, hu2023spontaneous, lian2023quantum, hu2024bosonic}.
However, note that obtaining a commutant generated by \textit{only} the Ising SZM $\Psi_{\ell}(q)$ requires generators that are strictly localized near the right edge, such as the above-mentioned $X_{L-1} Z_L$.
In fact, it is possible to show that there are no bulk local terms that can be added to the bond algebra that commute with $\Psi_\ell(q)$ but do not commute with $Q_Z$. 
This is due to the fact that for any local contiguous patch of sites $\{j, j+1, \cdots, j+n\}$, the ``Schmidt vectors" of $\Psi^\ell(q)$ on that patch include $\prod_{k = j}^{j+n}Z_n$, hence any local operator on those sites that commutes with $\Psi^\ell(q)$ also commutes with $Q_Z$. 
We also find the same to be the case for the YZ limit of the Fendley SZM we will discuss in Sec.~\ref{subsec:fendleynonint}.
All this suggests that any local Hamiltonian on a semi-infinite chain that has any of these SZMs on the open end necessarily also has a $Z_2$ symmetry, but this is important to prove more carefully in future work.\footnote{The fact that strictly local terms satisfy such properties does not necessarily imply that sums of strictly local terms satisfy such properties, and they need to be shown separately, as known from examples in the context of quantum many-body scars~\cite{moudgalya2022exhaustive, gioia2025distinct, odea2026equalspacing}.}
\section{Fendley SZM in Spin-1/2 XYZ Chain}\label{sec:fendleySZM}
With the appreciation that SZMs can be understood in terms of commutant algebras and can occur in non-integrable models, we now revisit the celebrated SZM proposed by Fendley in \cite{fendley2016strong} in the spin-1/2 XYZ model with OBC, given by the Hamiltonian
\begin{align}
    H_{\xyz} &= \sum_{j = 1}^{L-1}{(J_x X_j X_{j+1} + J_y Y_j Y_{j+1} + J_z Z_j Z_{j+1})}\nn \\
    &= \sum_{j = 1}^{L-1}{(x X_j X_{j+1} + y Y_j Y_{j+1} + Z_j Z_{j+1})},
\label{eq:XYZhamil}
\end{align}
where we have defined $x = J_x/J_z$ and $y = J_y/J_z$, and w.l.o.g.\ assumed $J_z = 1$. 
This Hamiltonian has two independent $\mathbb{Z}_2$ symmetries, generated by the operators 
\begin{equation}
Q_Z \defn \prod_{j=1}^{L} Z_j, \;\; Q_X \defn \prod_{j=1}^{L} X_j,
\label{eq:XYZsym}
\end{equation}
Reference~\cite{fendley2016strong} explicitly constructed an operator $\Psi(x, y)$ [precise expression given below] that satisfies $\norm{[\Psi(x,y), H_{\textrm{XYZ}}]} \sim \exp(-c L)$, i.e., the commutator with the Hamiltonian vanishes exponentially in the $L \rightarrow \infty$ limit, while $\norm{\Psi(x,y)}$ remains finite for $|x|,|y| < 1$.
Moreover, it satisfies $\{\Psi(x,y), Q_Z\} = 0$, leading to (almost) double degeneracies in the spectrum of the $H_{\xyz}$ for large system sizes. 
This is hence an example of an \textit{approximate} SZM in our definition, given in Sec.~\ref{sec:SZM}.

However, as mentioned in \cite{fendley2016strong}, it becomes an \textit{exact} SZM under the addition of an extra site $L+1$ and a boundary term $J_z Z_L Z_{L+1}$.
It is a bit more convenient to refer to the total system size as $L$, and hence $\Psi(x,y)$ \textit{exactly} commutes with the Hamiltonian
\begin{align}
    \widetilde{H}_{\textrm{XYZ}} &\defn H_{\textrm{XYZ}} -  x X_{L-1} X_{L} - y Y_{L-1} Y_L 
    ~.
\label{eq:HtildXYZ}
\end{align}
This ``fine-tuning" of the boundary term for the existence of an \textit{exact} SZM is reminiscent of SZMs can be understood in terms of bond algebras generated by a superposition of terms, such as those discussed in Secs.~\ref{sec:Ising} and \ref{sec:XYSZM}.
However, as we will show in Sec.~\ref{subsec:fendleynonint}, this understanding only goes through in the non-interacting limit of the XYZ model, where at least one of $J_x, J_y, J_z$ is $0$.  
Nevertheless, this analysis leads to interesting insights into the Fendley SZM for the interacting XYZ case.
\subsection{YZ limit}\label{subsec:fendleynonint}
To study the relation to bond algebras such as those in Sec.~\ref{sec:SZM}, we first consider a non-interacting limit of the XYZ model obtained by setting $J_x = 0$ (i.e., $x = 0$) in Eqs.~(\ref{eq:XYZhamil}) and (\ref{eq:HtildXYZ}), and we refer to the resulting Hamiltonians as $H_{\textrm{YZ}}$ and $\widetilde{H}_{\textrm{YZ}}$ respectively. 
Note that we choose the ``YZ limit" instead of the more standard ``XY limit" in order to stick with the notations of Ref.~\cite{fendley2016strong}, and distinguish this from the SZM in Sec.~\ref{subsec:XYSZM} that occurs in the presence of magnetic fields.
In the YZ limit, the SZM reads
\begin{equation}
    \Psi(y) \defn \Psi(0, y) = \sumal{s = 0}{\lfloor\frac{L-1}{2}\rfloor}{(-y)^s Z_{2s+1} \prodal{j = 1}{2s}{X_j}}.
\label{eq:Psidefn}
\end{equation}
Reference~\cite{fendley2016strong} noted that $\Psi(y)$ commutes exactly with the Hamiltonian $\tH_{\textrm{YZ}}$, which can be rewritten as
\begin{align}
    \tH_{\textrm{YZ}} &\defn \sumal{j = 1}{L-2}{\left(y Y_j Y_{j+1} + Z_j Z_{j+1}\right)} + Z_{L-1} Z_{L} \nn \\
    &= Z_1 Z_2 + \sumal{j = 1}{L-2}{\left(y Y_j Y_{j+1} + Z_{j+1} Z_{j+2}\right)}.
\label{eq:HYZ}
\end{align}
As we discuss below, $\Psi(y)$ commutes with the individual terms in the sum in the second line in Eq.~(\ref{eq:HYZ}).
\subsubsection{Matrix Product Operator (MPO) expression}
In order to see this, we express $\Psi(y)$ as an MPO using ideas in \cite{crosswhite2008fsa}, which reads
\begin{equation}
    \Psi(y) = \sum\limits_{\{s_n\}, \{t_n\}}{[{b^l}^T A_1^{[s_1 t_1]} A_2^{[s_2 t_2]} \dots A_L^{[s_L t_L]}  b^r]}\ket{\{s_n\}}\bra{\{t_n\}},
\label{eq:generalOBCMPO}
\end{equation}
where $A$'s can be thought of as $\chi \times \chi$ matrices (here $\chi = 3$) with elements expressed as $d \times d$ matrices (operators) acting on the physical indices (here $d = 2$), and $b^l$ and $b^r$ are $\chi$-dimensional boundary vectors of the MPO in the auxiliary space, which are set to $b^l = (1 \;\; 0 \;\; \cdots \;\; 0)^T$ and $b^r = (0 \;\; \cdots \;\; 0 \;\; 1)^T$ respectively.
The MPO tensor explicitly reads
\begin{equation}
    A_j = 
    \begin{pmatrix}
    0 & f X_j & Z_j \\
    f X_j & 0 & 0 \\
    0 & 0 & \mathds{1}
    \end{pmatrix},\;\;\;f \defn \sqrt{-y}.
\label{eq:psiMPO}
\end{equation}
In order to show that the MPO Eq.~(\ref{eq:generalOBCMPO}) commutes with the terms $\{y Y_j Y_{j+1} + Z_{j+1} Z_{j+2}\}$ in Eq.~(\ref{eq:HYZ}),  it is useful to obtain the three-site MPO $A_j A_{j+1} A_{j+2}$,  which reads
\begin{gather}
    A_j A_{j+1} A_{j+2} =\nn \\
    \begin{pmatrix}
    0 & f^3 X_j X_{j+1} X_{j+2} & Z_j + f^2 X_j X_{j+1} Z_{j+2}\\
    f^3 X_j X_{j+1} X_{j+2} & 0 & f X_j Z_{j+1} \\
    0 & 0 & \mathds{1}
    \end{pmatrix},
\label{eq:3siteMPO}
\end{gather}
and it is easy to verify that all the elements in Eq.~(\ref{eq:3siteMPO}) commute with $Z_{j+1} Z_{j+2} + y Y_j Y_{j+1}$,  hence we have $[Z_{j+1} Z_{j+2} + y Y_j Y_{j+1}, \Psi(y)] = 0$.
To show that $\Psi(y)$ commutes with $Z_1 Z_2$,  we can consider the two-site MPO on the leftmost sites along with the boundary vector $b^l$, we obtain
\begin{equation}
    {b^l}^T A_1 A_2 = 
    \begin{pmatrix}
    f^2 X_1 X_2 & 0 & Z_1
    \end{pmatrix}.
\label{eq:2siteMPOleft}
\end{equation}
All the elements in Eq.~(\ref{eq:2siteMPOleft}) commute with $Z_1 Z_2$, hence we have $[Z_1 Z_2, \Psi(y)] = 0$.
\subsubsection{Bond and Commutant algebras}
In all, these results show that $\Psi(y)$ is in the commutant $\mC_{\textrm{YZ}}$ of the following bond algebra:
\begin{equation}
    \mA_{\textrm{YZ}} = \lgen Z_1 Z_2, \{y Y_j Y_{j+1} + Z_{j+1} Z_{j+2}\} \rgen.
\label{eq:YZSZMbondalgebra}
\end{equation}
We numerically find that the commutant dimension is $\dim(\mC_{\textrm{YZ}}) = 16$, which can be understood as
\begin{equation}
    \mC_{\textrm{YZ}} = \lgen Q_X, Q_Z, \Psi(y), Z_{L} \rgen, 
\label{eq:CSZMexhaust}
\end{equation}
hence it is completely exhausted by the two $\mathbb{Z}_2$ symmetries $Q_X$ and $Q_Z$ of Eq.~(\ref{eq:XYZsym}), and the two SZMs $\Psi(y)$ and $Z_{L}$.
$\mC_{\textrm{YZ}}$ has a similar flavor as the Ising SZM commutant of Eq.~(\ref{eq:ISZMcomm}), with an extra $\mathbb{Z}_2$ symmetry due to absence of a magnetic field. 
Given this understanding, it is also easy to construct non-integrable models with $\Psi(y)$ as an SZM using ideas similar to those discussed in Sec.~\ref{subsec:nonintSZM}.
Moreover, the SZM can also be preserved while breaking some of the other symmetries in $\mC_{\textrm{YZ}}$, e.g., we can obtain the bond and commutant algebra pair
\begin{gather}
    \widetilde{\mA}_{\textrm{SZM}} = \lgen Z_1 Z_2, \{y Y_j Y_{j+1} + Z_{j+1} Z_{j+2}\}, \{Z_j Y_{j+1}\} \rgen,\nn \\
    \widetilde{\mC}_{\textrm{SZM}} = \lgen \Psi(y), Q_X\rgen.
\label{eq:YZSZMgen}
\end{gather}
\subsubsection{Properties of the SZM}
It is straightforward to show that $\Psi(y)$ satisfies the properties to be an \textit{exact} SZM according to the definitions of Sec.~\ref{sec:SZM}.
Moreover, for $|y|< 1$, it is localized at the left end of the chain.
The commutator of the normalized SZM with the full Hamiltonian $H_{\textrm{YZ}}$ [analogous to Eq.~(\ref{eq:approxSZMcondition})] obtained by setting $J_x = 0$ in Eq.~(\ref{eq:XYZhamil}) reads
\begin{equation}
    \frac{\norm{[H_{\textrm{YZ}}, \Psi(y)]}}{\norm{\Psi(y)}} = \frac{2 J_z |y|^{\floor{\frac{L-1}{2}}} \sqrt{1 - y^2}}{\sqrt{(1 - y^{2(\floor{\frac{L-1}{2}}+1)})}} \delta_{L, \textrm{even}}.
\end{equation}
Hence for $|y| <1$, $\Psi(y)$ is an \textit{approximate} SZM of $H_{\textrm{YZ}}$ for even $L$, while it remains an exact SZM for odd $L$.
For even $L$ and $|y| < 1$ (hence in the $z$-ferromagnetic phase), this leads to an approximate double degeneracy of the entire spectrum with exponentially small splittings (which are also smaller than the typical level spacing if $|y|$ is sufficiently small).
On the other hand, for odd $L$, this leads to an exact double degeneracy of the entire spectrum for all values of $y$ (hence in both the $z$-ferromagnetic and $y$-ferromagnetic phases).
However, in that case this exact degeneracy can also be explained simply using the anticommutation of $Q_X$ and $Q_Z$ for those system sizes, and the SZM $\Psi(y)$ does not lead to any additional degeneracies. 
\subsection{XYZ Chain}\label{subsec:fendleyint}
Armed with the understanding of the SZM of $H_{\rm YZ}$, we now move on to $H_{\rm XYZ}$ when $x \neq 0$.
Here we find that there \textit{cannot} be an analogous commutant understanding of the SZM $\Psi(x,y)$. 
Nevertheless, similar ideas lead to an alternate proof for the SZM. 
\subsubsection{Matrix Product Operator Expression}\label{subsec:MPOexpression}
Given the utility of the MPO of Eq.~(\ref{eq:generalOBCMPO}) in the YZ limit, we first recast the Fendley SZM as a Matrix Product Operator (MPO) with a bond dimensions of $\chi = 4$ as
\begin{gather}
    \Psi(x,y) = \sum\limits_{\{s_n\}, \{t_n\}}{[{b_F^l}^T F_1^{[s_1 t_1]} F_2^{[s_2 t_2]} \dots} \nn\\
    {F_L^{[s_L t_L]} b_F^r]}\ket{\{s_n\}}\bra{\{t_n\}},
\label{eq:PsiLgeneral}
\end{gather}
where $(b_F^l)^T = \begin{pmatrix}1 & 0 & 0 & 0\end{pmatrix}$,  $(b_F^r)^T = \begin{pmatrix}0 & 0 & 0 & 1\end{pmatrix}$, are the boundary vectors, and \begin{equation}
    F_j = 
    \begin{pmatrix}
    x y \mathds{1}_j & f X_j & g Y_j & Z_j \\
    f X_j & x\mathds{1}_j  & 0 & 0 \\
    g Y_j & 0 & y\mathds{1}_j & 0 \\
    0 & 0 & 0 & \mathds{1}_j
    \end{pmatrix},\;\;\begin{array}{c}
    f \defn \sqrt{y (x^2 - 1)}\\
    g \defn \sqrt{x(y^2 - 1)}\end{array},
\label{eq:XYZSZMMPO}
\end{equation}
which can be derived by staring at the expressions in~\cite{fendley2016strong} and using standard methods for constructing MPOs~\cite{crosswhite2008fsa}.
As standard with MPOs, this expression is not unique and there are many degrees of freedom, such as the coefficient in front of $Z_j$, which only multiplies $\Psi(x,y)$ by an overall constant, which we have set to $1$.
There is also a freedom of multiplying one of the $X_j$ (or $Y_j$) by a c-number, and dividing the other $X_j$ (or $Y_j$) by the same c-number; choosing it to be $i$ makes the MPO manifestly Hermitian for $|x|, |y| < 1$ that can be convenient for some computations.
Equivalent expressions were also independently obtained by Fendley \textit{et.~al.}~\cite{fendleycomm, fendley2025xyz} [which can be obtained by multiplying  $F_j$ of Eq.~(\ref{eq:XYZSZMMPO}) by an overall factor of $x^{-1}y^{-1}$ and setting the coefficient of $Z_j$ to $1$ by noting the freedom there], and by Katsura~\cite{katsuracomm}.
This MPO expression by itself simplifies the study of the properties of the Fendley SZM.
For example, its localization on the left end of the chain for $|x|, |y| < 1$ can be directly read off to be a consequence that the various coefficients $xy$, $f$, $g$, $x$, $y$ that appear in the MPO, all of which have absolute values less than $1$ when $|x|, |y| < 1$. 
Furthermore, the normalizability and the Majorana nature in the thermodynamic limit can also be proven using the structures of appropriate transfer matrices of these MPOs, as demonstrated in App.~\ref{app:fendleysimplified}.
Finally, the fact that this MPO commutes with the Hamiltonian $\tH_{\rm XYZ}$ can be shown using a telescoping cancellation of terms, as discussed below and as also recently demonstrated in Ref.~\cite{fendley2025xyz}. 
\subsubsection{Telescoping Structure}
To show that this MPO commutes with $\tH_{\rm XYZ}$, we first rewrite it in a manner similar to Eq.~(\ref{eq:HYZ}), as
\begin{align}
    \tH_{\rm XYZ} &=  \sumal{j = 1}{L-2}{\left(x X_j X_{j+1} + y Y_j Y_{j+1} + Z_j Z_{j+1}\right)} + Z_{L-1} Z_{L} \nn \\
    &= Z_1 Z_2 + \sumal{j = 1}{L-2}{\left(x X_j X_{j+1} + y Y_j Y_{j+1} + Z_{j+1} Z_{j+2}\right)}.
\label{eq:HXYZ}
\end{align}
We then consider the commutation of the three-site term
\begin{equation}
    T_{j,j+1,j+2} \defn x X_j X_{j+1} + y Y_j Y_{j+1} + Z_{j+1} Z_{j+2}
\label{eq:threesitegroup}
\end{equation} 
that appears in Eq.~(\ref{eq:HXYZ}) with the elements of the three-site MPO $F_j F_{j+1} F_{j+2}$. 
As shown in Sec.~\ref{subsec:fendleynonint}, this perfectly commutes when $x = 0$ (or, equivalently, when $y = 0$). 
However, when $x \neq 0$ and $y \neq 0$, we find that this no longer commutes. 
Instead, the fact that $\tH_{\rm XYZ}$ commutes with $\Psi(x,y)$ comes from a more intricate telescoping structure that we discuss in App.~\ref{subsec:altproof} and summarize below. 
In particular, we find that
\begin{align}
    [T_{j,j+1,j+2}, \Psi(x,y)] &= -\widehat{\Psi}_{j,j+1}(x,y) +  \delta_{j \neq L-2}\widehat{\Psi}_{j+1,j+2}(x,y),\nn \\
    [Z_{1} Z_{2}, \Psi(x,y)] &= \widehat{\Psi}_{1,2}(x,y),
\label{eq:telescopingcomm}
\end{align}
where $\widehat{\Psi}_{j,j+1}(x,y)$ has an interpretation of inserting a certain ``defect" into $\Psi(x,y)$ on sites $j$ and $j+1$, and the precise expression is given in Eq.~(\ref{eq:Psihatdefn}) in App.~\ref{app:fendleysimplified}. 
The form of the Hamiltonian in Eq.~(\ref{eq:HXYZ}), along with the relations of Eq.~(\ref{eq:telescopingcomm}), show that 
\begin{equation}
[\tH_{\textrm{XYZ}}, \Psi(x,y)] = 0.
\end{equation}
While this relation is true for all values of $x$ and $y$, for $|x|, |y| < 1$ we can use the fact that $\Psi(x,y)$ is exponentially localized on the left end of the chain to also show that 
\begin{align}
    \norm{[H_{\rm XYZ}, \Psi(x,y)]} &= \norm{[x X_{L-1} X_L + y Y_{L-1} Y_L, \Psi(x,y)]}\nn \\
    &\sim \exp(-c L),\;\;\;\text{if}\;\;\;|x|, |y| < 1.
\label{eq:Fendleynormexp}
\end{align}
This can also be shown more systematically using MPO transfer matrices, as discussed in App.~\ref{subsec:approxSZM}.
Note that this method of telescoping series has been used to demonstrate the existence of several examples of SZMs in the literature~\cite{vasiloiu2019strong}, and more recently also the Fendley SZM itself~\cite{fendley2025xyz}.\footnote{Our proof differs from the one in Ref.~\cite{fendley2025xyz} in that we use three-site terms grouped as in Eq.~(\ref{eq:threesitegroup}) whereas Ref.~\cite{fendley2025xyz} directly uses the nearest-neighbor terms of the XYZ model. The three-site grouping used here naturally reveals the connections to commutants in the non-interacting limit, as we discuss in Sec.~\ref{subsec:fendleynonint}.}
What makes $\Psi(x,y)$ distinct from those examples is that the telescoping is \textit{necessary}, as we discuss in the next subsection, whereas in many other cases such as those in Secs.~\ref{sec:Ising} and \ref{subsec:fendleynonint}, the telescoping can be eliminated by a suitable choice of groupings of terms.
\subsubsection{Necessity of the Telescoping Structure}
We now discuss why the telescoping structure of Eq.~(\ref{eq:telescopingcomm}) is necessary, and why there cannot be any regrouping of the terms of Eq.~(\ref{eq:HXYZ}) such that the individual parts commute with $\Psi(x,y)$. 
In other words, we argue that there is no non-trivial bond algebra generated by strictly local terms such that $\Psi(x,y)$ lies in its commutant. 
First, we can also explicitly search for strictly local operators (i.e., generators of a potential bond algebra) that commute with $\Psi(x,y)$, and we find that there are no such operators of range $\leq 6$, as discussed in App.~\ref{app:fendleysimplified}.
More generally, this can also be understood using the difference between the spectra of $\Psi(y)$ and $\Psi(x,y)$. 
$\Psi(y)$, upon proper normalization assumed here and below, has a highly degenerate spectrum with $2^{L-1}$ eigenvalues being $+1$ and the other $2^{L-1}$ being $-1$, consistent with $\Psi(y)^2 \propto \mathds{1}$.
This degeneracy of the spectrum is important for assigning only two distinct quantum numbers to the spectrum of any term or Hamiltonians with this symmetry, the quantum numbers are simply the distinct eigenvalues of the symmetry operator. 
This degeneracy also allows the possibility for the existence of multiple terms that commute with $\Psi(y)$ but do not commute among themselves, e.g., the terms that generate the bond algebra $\mA_{\rm YZ}$ of Eq.~(\ref{eq:YZSZMbondalgebra}). 
\begin{figure}
\includegraphics[scale=0.95]{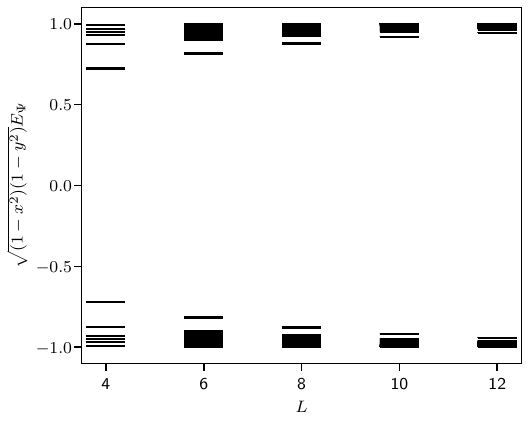}
\caption{\textbf{Eigenvalues of the Fendley SZM $\boldsymbol{\Psi(x,y)}$.}
Data shown for $(x, y) = (0.6, 0.4)$ for various system sizes from $L = 4$ to $L = 12$. 
After scaling by the inverse norm $\norm{\Psi(x,y)}^{-1} = \sqrt{(1 - x^2)(1 - y^2)}$, the eigenvalues $\{E_\Psi\}$ are increasingly clustered around $+1$ and $-1$ for increasing system size $L$, but are in fact completely non-degenerate for any finite system size $L$.
This is to be contrasted with other SZMs discussed in this work, where all eigenvalues of the SZM at finite system size are either $+1$ or $-1$. 
}
\label{fig:Fendleyeig}
\end{figure}
On the other hand, the spectrum $\Psi(x,y)$ is completely \textit{non-degenerate} for any finite system size $L$, as shown in Fig.~\ref{fig:Fendleyeig}. 
After proper normalization, it contains two clusters of eigenvalues: $2^{L-1}$ of them are close to $+1$, and $2^{L-1}$ of them are close to $-1$, which is consistent with the fact that $\Psi(x,y)^2$ is not proportional to $\mathds{1}$ for finite sizes, but only in the $L \rightarrow \infty$ limit~\cite{fendley2016strong} (see App.~\ref{app:fendleysimplified} for a proof of this using the MPO expression). 
This non-degenerate spectrum of $\Psi(x,y)$ implies that all terms that commute with it are diagonal in the eigenbasis of $\Psi(x,y)$ and hence should also commute among themselves. 
This precludes the existence of a bond algebra generated by non-commuting terms.
In fact, for any finite system size, this also guarantees that any operator that commutes with $\Psi(x,y)$ can be generated (in the algebra sense, i.e., as a linear combination of powers) from $\Psi(x,y)$ [although formally all powers of $\Psi(x,y)$ up to $2^L-1$ are needed, so this is not practically very useful].
Since $\tH_{\rm XYZ}$ commutes with $\Psi(x,y)$, it should also be expressible in terms of $\Psi(x,y)$ for any finite system size.
This shows that telescoping structures such as Eq.~(\ref{eq:telescopingcomm}) are necessary under any regrouping of terms of $\tH_{\rm XYZ}$.
Indeed, the existence of symmetry operators with non-degenerate spectra is characteristic of Bethe Ansatz integrable systems, hence we suspect that the existence of $\Psi(x,y)$ is tied to the Bethe Ansatz integrability of $H_{\textrm{XYZ}}$.
This is consistent with the recent observation of Ref.~\cite{fendley2025xyz} that all the integrals of motion of the XYZ model can be derived from a general class of MPOs that commute with that Hamiltonian, with the MPO equivalent to that of Eq.~(\ref{eq:XYZSZMMPO}) being a special case (see also Ref.~\cite{yashin2026MPO}).
\section{Conclusions and Outlook}\label{sec:conclusions}
In this work, we first introduced a technique to search for interesting symmetries or commutant algebras that are realizable in certain families of local Hamiltonians.
This utilized the idea that commutant algebras can be understood as ground states of local ``super-Hamiltonians," which were introduced in earlier numerical~\cite{moudgalya2023numerical} and analytical~\cite{moudgalya2023symmetries} works. 
Doing such computations for very small system sizes (e.g., $L = 4$ for spin-1/2 systems) and only on a small class of bond algebras generated by nearest-neighbor terms already led to finding novel unconventional symmetries that can exist in local Hamiltonians.
We find it remarkable that all these symmetries, some of which have arisen in very different contexts and have taken non-trivial effort to be identified in the past body of work, can in principle be ``discovered" using simple brute-force searches demonstrated in this paper. 
It would be interesting to apply this technique to larger families of Hamiltonians and potentially coarsely ``classify" the kinds of symmetries that can occur in such systems.
However, given the large variety of symmetries that can occur in higher spin systems or longer-range spin-1/2 systems, such as Hilbert space fragmentation~\cite{moudgalya2021hilbert}, quantum many-body scars~\cite{moudgalya2022exhaustive}, modulated symmetries~\cite{sala2022dynamics}, or even categorical symmetries~\cite{chatterjee2024quantum, bhardwaj2024lattice, bhardwaj2024illustrating}, this might be a formidable task, and the fact that the search space grows drastically also warrants better search methods. 
A simpler version of this problem might be to consider bond algebras generated by Pauli strings or ``stabilizer terms" and to classify their commutants.
This is along the lines of the recent classification of so-called Dynamical Lie Algebras~\cite{wiersema2023classification} generated by Pauli strings, which are closely related to superoperator symmetries~\cite{zeier2011symmetry, zimboras2015symmetry, lastres2024nonuniversality}.
While the complete set of symmetries is quite large and will be explored in more detail elsewhere, in this work we focused on Strong Zero Modes (SZMs), where the understanding in terms of commutants clarifies a number of aspects. 
First, this clearly shows that exact SZMs can exist in \textit{non-integrable} models, which was unclear from previous works on this subject. 
Second, it clarifies the relation between the SZMs and ground state phases of matter by revealing clear conditions under which they become exact or approximate conserved quantities.   
Third, this also reveals novel quasilocal $U(1)$ symmetries that can exist alongside certain SZMs, such as in the spin-1/2 XY model for a particular choice of parameters. 
Fourth, it reveals the possibility of apparent hydrodynamic behavior that should exist in non-integrable models with SZMs, such as in the critical Ising model with appropriate interactions, or those with the novel quasilocal $U(1)$ symmetries, which would all be interesting to numerically explore carefully in future work.
Finally, this motivated us to revisit the celebrated Fendley SZM~\cite{fendley2016strong} in the spin-1/2 XYZ model and check its connections to the commutant framework. 
We found that the Fendley SZM can be understood in terms of commutants in the non-interacting limit, but not in the interacting case.  
Nevertheless, this led to an alternate proof for the Fendley SZM, similar to what was recently shown in Ref.~\cite{fendley2025xyz}.
This suggests the existence of two kinds of SZMs, ones that can be extended to non-integrable models, and ones that are more tightly connected to integrability. 
In the light of this understanding of SZMs, it would be interesting to revisit the many other examples of SZMs in the literature~\cite{richelli2024brick, znidaric2024integrability, olund2023boundary, vernier2024strong, essler2025szm, gehrmann2025exact} and ask if integrability is necessary for their existence.
It would also be interesting to check if the commutant understanding goes through to SZMs in parafermion systems~\cite{fendley2012parafermionic, jermyn2014stability, alicea2016topological}, to SZMs in Floquet systems~\cite{yates2019almost, mukherjee2023emergent, vernier2024strong},  where there are also many other notions of long-lived modes beyond SZMs~\cite{yates2019almost, yates2021strong, yeh2024productmode}, and finally if there are some implications to classical Markov processes with SZMs~\cite{klobas2023stochastic}.  
Last but not least, while we have focused on systems in one spatial dimension, it would also be interesting to explore SZMs~\cite{vasiloiu2019strong, chertkov2020engineering, wildeboer2021symmetryprotected, kuno2023interplay, lehmann2023fragmentation} or their analogues in higher dimensions.

\section*{Acknowledgements}
We thank Bryan Clark, Paul Fendley, Hosho Katsura, and Aditi Mitra for useful discussions.
This work is supported by the Walter Burke Institute for Theoretical Physics at Caltech; the Institute for Quantum Information and Matter, an NSF Physics Frontiers Center (NSF Grant PHY-1733907); the National Science Foundation through grant DMR-2001186; and the Munich Center for Quantum Science and Technology (MCQST), which is supported by Germany’s Excellence Strategy--EXC--2111--390814868.
\bibliography{newrefs2}
\appendix 
\onecolumngrid
\section{Quasi-local \texorpdfstring{$U(1)$}{} Symmetries with SZMs}\label{app:hiddenU1}
In this appendix, we provide additional details on the quasi-local $U(1)$ symmetries discussed in Sec.~\ref{subsec:quasilocalU1I}.
As discussed in the main text, we find that the symmetry operators are given by
\begin{gather}
    \mathcal{Q}_I^X(q)  \defn \sum_{j=1}^L{Z_j} - \eta \sum_{1 \leq j_1 < j_2 \leq L}{X_{j_1} \left(\prod_{i = j_1+1}^{j_2 - 1} q Z_i\right) X_{j_2}}, \nn\\
    \mathcal{Q}_I^Y(q) \defn \sum_{j=1}^L{Z_j} + \eta \sum_{1 \leq j_1 < j_2 \leq L}{Y_{j_1} \left(\prod_{i = j_1+1}^{j_2 - 1} \frac{1}{q} Z_i\right) Y_{j_2}},
\label{eq:nonlocalU1}
\end{gather}
where $\eta \defn q - q^{-1}$.
To show that these commute with the generators of the bond algebra $\mA_{U(1),I}$ of Eq.~(\ref{eq:U1bondI}), it is convenient to express them as Matrix Product Operators (MPOs),  as
\begin{equation}
\mQ = \sum\limits_{\{s_n\}, \{t_n\}}{[{b^l}^T A_1^{[s_1 t_1]} A_2^{[s_2 t_2]} \dots A_L^{[s_L t_L]} b^r]}\ket{\{s_n\}}\bra{\{t_n\}},
\label{eq:generalOBCMPOnonloc}
\end{equation}
where $A$ can be thought of as $\chi \times \chi$ matrices with elements being $d \times d$ matrices acting on the physical indices, and $b^l$ and $b^r$ are $\chi$-dimensional boundary vectors of the MPO in the auxiliary space, which are set to $b^l = (1 \;\; 0 \;\; \cdots \;\; 0)^T$ and $b^r = (0 \;\; \cdots \;\; 0 \;\; 1)^T$ respectively.
The MPO tensors for the conserved quantities of Eq.~(\ref{eq:nonlocalU1}) read
\begin{equation}
    A^X_j = 
    \begin{pmatrix}
    \mathds{1}_j & \sqrt{\eta} X_j & Z_j \\
    0 & q Z_j & -\sqrt{\eta} X_j \\
    0 & 0 & \mathds{1}_j
    \end{pmatrix},\;\;\;A^Y_j = 
    \begin{pmatrix}
    \mathds{1}_j & \sqrt{\eta} Y_j & Z_j \\
    0 & q^{-1} Z_j & \sqrt{\eta} Y_j \\
    0 & 0 & \mathds{1}_j
    \end{pmatrix}.
\label{eq:AXAYMPO}
\end{equation}
The two-site MPOs then read
\begin{gather}
	A^X_j A^X_{j+1} = 
	\begin{pmatrix}
	\mathds{1}_{j,j+1} & \sqrt{\eta} X_{j+1} + \sqrt{\eta} q X_j Z_{j+1} &   Z_{j+1} - \eta X_j X_{j+1} + Z_j \\
	0 & q^2 Z_j Z_{j+1} & - \sqrt{\eta} Z_j X_{j+1} - \sqrt{\eta} X_{j+1} \\
	0 & 0 & \mathds{1}_{j,j+1}
	\end{pmatrix}, \\
	A^Y_j A^Y_{j+1} = 
	\begin{pmatrix}
	\mathds{1}_{j,j+1} & \sqrt{\eta} Y_{j+1} + \sqrt{\eta} q Y_j Z_{j+1} &   Z_{j+1} + \eta Y_j Y_{j+1} + Z_j \\
	0 & q^2 Z_j Z_{j+1} & \sqrt{\eta} Z_j Y_{j+1} + \sqrt{\eta} Y_{j+1} \\
	0 & 0 & \mathds{1}_{j,j+1}
	\end{pmatrix}.
\end{gather}
All the operators in these two-site tensors can be shown to exactly commute with the generators of the bond algebra $\mA_{U(1), I}$, i.e., 
\begin{equation}
    [(A^{X/Y}_j A^{X/Y}_{j+1})^{[\mu \nu]}, q^{-1} X_j X_{j+1} + q Y_j Y_{j+1} + Z_j + Z_{j+1}] = 0,\;\;
\end{equation}
where the indices $\mu,\nu$ label the auxiliary indices after multiplication of the matrices in Eq.~(\ref{eq:AXAYMPO}).
This shows that the full MPO of Eq.~(\ref{eq:generalOBCMPOnonloc}) also commutes with each of the bond algebra generators, hence showing that $\mQ^X_{I}$ and $\mQ^Y_{I}$ are in the commutant $\mC_{U(1), I}$. 
A similar construction holds for the staggered quasi-local $U(1)$ symmetry of Sec.~\ref{subsubsec:staggeredquasilocalU1}.
\section{Translation of Results to the Fermionic Language}\label{app:jordanwigner}
We now translate many of our results to the fermionic language using the Jordan-Wigner transformation, which clarifies many aspects of the SZMs and other symmetry operators discussed in the main text.
We start by defining spinless fermion creation and annihilation operators $\{c^\dagger_j\}$ and $\{c_j\}$ that satisfy $\{c^\dagger_j, c_k\} = \delta_{j,k}$, $\{c^\dagger_j, c^\dagger_k\} = \{c_j, c_k\} = 0$, and we define the fermion number operator as $n_j = c^\dagger_j c_j$. 
We then perform the transformation~\cite{mbeng2024quantum},\footnote{Note that some references define this instead by $Y_j \rightarrow -Y_j$ and $Z_j \rightarrow -Z_j$, which is equivalent. We have made this choice to maintain elegance of some expressions below.}
\begin{equation}
    X_j = (-1)^{\sum_{k=1}^{j-1}{n_k}}(c_j^\dagger + c_j),\;\;\;Y_j = i(-1)^{\sum_{k=1}^{j-1}{n_k}}(c_j^\dagger - c_j),\;\;\;Z_j = 1 - 2n_j = (-1)^{n_j},\;\;\;
\end{equation}
which can be verified to satisfy the necessary commutation relations for Pauli matrices, e.g., $[X_j, Y_k] = 2 i Z_j \delta_{j,k}$.
We can further define $2L$ Majorana operators $\{\chi_{2j-1}, \chi_{2j}\}$ as 
\begin{equation}
    c_j^\dagger = \frac{1}{2}(\chi_{2j-1} + i \chi_{2j}),\;\;c_j = \frac{1}{2}(\chi_{2j-1} - i\chi_{2j}),
\label{eq:Majoranamode}
\end{equation}
which can be shown to satisfy the relations $\{\chi_a, \chi_b\} = 2\delta_{a,b}$, that ensures that all the fermionic anti-commutation relations are satisfied. 
In terms of the regular (a.k.a.\ complex) fermions and Majorana (a.k.a.\ real) fermions, the terms we study in the main text translate to
\begin{gather}
    X_j X_{j+1} =  c^\dagger_j c_{j+1} + c^\dagger_j c^\dagger_{j+1} + h.c. = i \chi_{2j}\chi_{2j+1},\;\;\;Y_j Y_{j+1} = c^\dagger_j c_{j+1} - c^\dagger_j c^\dagger_{j+1} + h.c. = -i\chi_{2j-1}\chi_{2j+2},\nn \\
    X_j Y_{j+1} = i(-c^\dagger_j c_{j+1} + c^\dagger_j c^\dagger_{j+1}) + h.c. = -i\chi_{2j}\chi_{2j+2},\;\;\;Y_j X_{j+1} = i(c^\dagger_j c_{j+1} + c^\dagger_j c^\dagger_{j+1}) + h.c. = i\chi_{2j-1}\chi_{2j+1},\nn \\
    Z_j = 1 - 2 n_j = i \chi_{2j-1} \chi_{2j}.
\label{eq:localopsfermion}
\end{gather}
This makes it clear that the generators of the bond algebras we study are all quadratic fermion operators.
Below we describe operators in various commutants in the fermionic language. 
The $\mathbb{Z}_2$ symmetry operator $Q_Z$,  which is common among all the commutants,  in the fermionic language reads
\begin{equation}
	Q_Z = (-1)^{\sum_{j = 1}^L{n_j}} = i^L\prod_{k = 1}^{2L}{\chi_k}.
\label{eq:Z2symfermion}
\end{equation}
\subsection{Ising SZM}
We start with the Ising bond algebra $\mA_{\iszm}$ of Eq.~(\ref{eq:Isingbond}) and its commutant $\mC_{\iszm}$ of Eq.~(\ref{eq:ISZMcomm}).
The generators of $\mA_{\iszm}$ read
\begin{equation}
	X_j X_{j+1} + q Z_j = i\chi_{2j} (-q \chi_{2j-1} + \chi_{2j+1}).
\label{eq:IsingMajorana}
\end{equation}
The Ising SZM of Eq.~(\ref{eq:ISZMdefns}) is also a Majorana Zero Mode (MZM), and reads
\begin{equation}
    \Psi^\ell(q) = \sum_{j = 1}^L{q^{j-1}(c^\dagger_j + c_j)}  = \sum_{j = 1}^L{q^{j-1} \chi_{2j-1}},
\label{eq:IsingSZMfermion}
\end{equation}
which commutes with the generators of $\mA_{\iszm}$ due to the anticommutation relation
\begin{equation}
    \{\chi_{2j-1} + q\chi_{2j+1}, q \chi_{2j-1} - \chi_{2j+1}\} = 0,
\end{equation}
which is analogous to Eq.~(\ref{eq:Isinglocal}).
$Q_Z \Psi^\ell(q)$ can then be obtained using Eqs.~(\ref{eq:Z2symfermion}) and (\ref{eq:IsingSZMfermion}), and is not an MZM, i.e., it is a product of multiple Majorana operators.
On the other hand,  while $X_L$ is also not an MZM,  and reads $X_L = i^{L-1} \prod_{k = 1}^{2L-2}\chi_k \times \chi_{2L-1}$, we can use it to build the MZM
\begin{equation}
i Q_Z X_L = \chi_{2L}.
\label{eq:rightmajorana}
\end{equation}
As discussed in Sec.~\ref{subsubsec:SPspectrum}, this pair of MZMs, $\Psi^\ell(q)$ and $i Q_Z X_L$, are also precisely the pair of MZMs that appear in the single-particle spectrum of any quadratic free-fermion Hamiltonian in the bond algebra $\mA_{\iszm}$.
When $|q| < 1$, they are localized on opposite ends of the chain.
This leads to the stability of the two-fold degeneracy of the spectrum in the presence of additional terms in that parameter regime.
This also coincides with the topological phase of the clean Kitaev chain that is the uniform superposition of the terms in Eq.~(\ref{eq:IsingMajorana}) [which, in the spin language, is the ferromagnetic phase of the transverse-field Ising model constructed as such].
\subsection{XY SZM}\label{subsec:XYSZMapp}
We can similarly translate the results for the XY SZM of Sec.~\ref{subsec:XYSZM} to the Majorana language. 
The generators read
\begin{gather}
    (\kappa + \gamma) X_j X_{j+1} + (\kappa - \gamma) Y_j Y_{j+1} + q (\kappa+\gamma)Z_j + q^{-1}(\kappa - \gamma) Z_{j+1} \nn \\
    = (\kappa + \gamma) i \chi_{2j} \chi_{2j+1} -i (\kappa - \gamma) \chi_{2j-1} \chi_{2j+2} + iq(\kappa + \gamma) \chi_{2j-1} \chi_{2j} + iq^{-1}(\kappa - \gamma) \chi_{2j+1} \chi_{2j+2} \nn \\
    = i(\kappa + \gamma) \chi_{2j} (\chi_{2j+1} - q \chi_{2j-1}) + i (\kappa - \gamma) \chi_{2j+2} (\chi_{2j-1} - q^{-1}\chi_{2j+1}).
\label{eq:XYtermsmajorana}
\end{gather}
As discussed in Sec.~\ref{sec:XYSZM}, the bond algebra generated by these terms has two SZMs in addition to the $Z_2$ symmetry $Q_Z$ of Eq.~(\ref{eq:Z2symfermion}). 
One of them is simply $\Psi^\ell(q)$ of Eq.~(\ref{eq:IsingSZMfermion}), and the other is $\Psi^r(s)$ defined in Eq.~(\ref{eq:PsiR}) for $s = q^{-1}\frac{\kappa - \gamma}{\kappa + \gamma}$.
Similar to $X_L$ in the Ising case, the SZM $\Psi^r(s)$ is not an MZM, but an MZM can be constructed using $Q_Z$ as
\begin{equation}
    i Q_Z \Psi^r(s) = -\sum_{j = 1}^L{s^{L-j}\left(\prod_{k = 1}^{j-1}{Z_k}\right) Y_j} = -i\sum_{j = 1}^L{s^{L-j} (\cd_j - c_j)} = \sum_{j = 1}^L{s^{L-j}\chi_{2j}}.
\end{equation}

Similar to the Ising case, this pair of MZMs, $\Psi^\ell(q)$ and $i Q_Z \Psi^r(s = q^{-1}\frac{\kappa - \gamma}{\kappa + \gamma})$, are precisely the MZMs that appear in free-fermion models constructed out of the terms given by Eq.~(\ref{eq:XYtermsmajorana}). 
In addition, as we show below, these MZMs are localized on opposite ends of the chain when the translation-invariant Majorana chain built out of such terms is in its topological phase in its bulk [or equivalently, the XY chain is in its ferromagnetic phase].
\begin{lemma}
\label{lem:XYlemma}
Let $\kappa,\gamma,q\in\mathbb{R}$ with $\kappa+\gamma\neq 0$, $\kappa - \gamma \neq 0$, $q^{-1} \neq 0$, and $q\neq 0$.
The MZMs $\Psi^\ell(q)$ and $i Q_Z\Psi^r(s = q^{-1}\frac{\kappa - \gamma}{\kappa + \gamma})$ of Eqs.~(\ref{eq:SZMMajorana}) and (\ref{eq:Psirmajorana}) are localized on opposite ends of the chain if and only if the bulk of the XY chain of Eq.~(\ref{eq:XYHamil}) is in a ferromagnetic phase, which is known to occur when
\begin{equation}
    \bigl|\,q(\kappa+\gamma)+q^{-1}(\kappa-\gamma)\,\bigr| < 2|\kappa|.
\label{eq:XYFMphase}
\end{equation}
\end{lemma}
\begin{proof}
Note that the said MZMs are localized on opposite ends of the chain in either of the following two cases:
\begin{equation}
  \text{(i)}\;\;|q|<1\;\; \text{and}\;\; \left|\,q^{-1}\frac{\kappa-\gamma}{\kappa+\gamma}\right|<1,\;\;\;\;\;\text{(ii)}\;\;|q|>1\;\;\text{and}\;\;\left|\,q^{-1}\frac{\kappa-\gamma}{\kappa+\gamma}\right|>1.
\label{eq:localizationcases}
\end{equation}
To show that these cases are equivalent to  Eq.~(\ref{eq:XYFMphase}), we first define
\begin{equation}
s := q^{-1}\frac{\kappa-\gamma}{\kappa+\gamma},
\end{equation}
so that $\kappa-\gamma = (\kappa+\gamma)qs$.
Then we have the following relations 
\begin{equation}
|q(\kappa+\gamma)+q^{-1}(\kappa-\gamma)|
  = |(\kappa+\gamma)(q+s)|,\;\;2|\kappa| 
        = |(\kappa+\gamma)(1+qs)|.
\end{equation}
Thus, the equivalence between Eqs.~(\ref{eq:XYFMphase}) and (\ref{eq:localizationcases}) can directly be established as
\begin{equation}
\bigl|\,q(\kappa+\gamma)+q^{-1}(\kappa-\gamma)\,\bigr|<2|\kappa|
\;\Longleftrightarrow\;
|q+s|<|1+qs|\;\Longleftrightarrow (1-q^2)(1-s^2) > 0\;\Longleftrightarrow |q|, |s| < 1\; \text{or}\;|q|,|s| >1,
\label{eq:toshowsimp}
\end{equation}
where in the second to last step we have simply squared both sides and rewritten the inequality.
This completes the proof.
\end{proof}

\subsection{Conventional and Quasilocal \texorpdfstring{$U(1)$}{} Symmetries}\label{subsec:U1symmetries}
We now mention the Majorana forms of the conventional and quasilocal $U(1)$ symmetries discussed in Secs.~\ref{subsec:conventionalU1} and \ref{subsec:quasilocalU1I}; both are quadratic in the Majorana language.
The conventional $U(1)$ symmetry is given by the conserved total spin
\begin{equation}
    \sum_j{Z_j} = i\sum_{j = 1}^{L}{\chi_{2j-1}\chi_{2j}}.
\end{equation}
In order to contrast this with the quasilocal $U(1)$ symmetries discovered here, we briefly discuss the single-particle spectrum of this operator, since it is a quadratic Majorana operator. 
Since it is just a sum of commuting terms, the single-particle spectrum, which has $2L$ levels, just consists of $L$-fold degenerate $-1$s and $L$-fold degenerate $+1$s. 
This gives rise to the standard many-body spectrum of this total $U(1)$ charge operator that consists of $2^L$ energy levels.
This is obtained by always filling in all the $L$ of $-1$ single-particle levels, populating the remaining $L$ of $+1$ single-particle levels in $2^L$ different ways.
The many-body eigenvalues can be directly computed in terms of the single-particle eigenvalues. 
Given the single-particle eigenvalues $\{\pm \varepsilon_\alpha\}$, the many-body eigenvalues are given by $\{\sum_{\alpha =1}^{L}{(2 n_\alpha -1)\varepsilon_\alpha}\}$, where $n_\alpha \in \{0, 1\}$ is the occupation of the single-particle level with energy $+\varepsilon_\alpha$.
As easy to compute, this gives rise to many-body eigenvalues $\{-L + 2n\}$ for $0 \leq n \leq L$, a total of $L+1$ distinct eigenvalues, as is well-known. 
Below we will see that this is different for the quasilocal $U(1)$ operator. 
The quasi-local $U(1)$ symmetry operator $\mQ^X_I(q)$ reads
\begin{align}
    \mathcal{Q}^X_I(q) 
    &= \sum_{j=1}^L{Z_j} - \eta \sum_{1 \leq j_1 < j_2 \leq L}{q^{j_2 - j_1 - 1}X_{j_1} \left(\prod_{k = j_1+1}^{j_2 - 1} Z_k\right) X_{j_2}} \nn \\
    &= \sum_{j=1}^L{(1 - 2 n_j)} - \eta \sum_{1 \leq j_1 < j_2 \leq L}{q^{j_2 - j_1 - 1} (-1)^{\sum_{k = 1}^{j_1 -1}n_k}(\cd_{j_1} + c_{j_1}) (-1)^{\sum_{k = j_1+1}^{j_2 -1}n_k}(-1)^{\sum_{k = 1}^{j_2 -1}n_k} (\cd_{j_2} + c_{j_2})}\nn \\
    &= \sum_{j=1}^L{(1 - 2 n_j)} - \eta \sum_{1 \leq j_1 < j_2 \leq L}{q^{j_2 - j_1 - 1} (\cd_{j_1} - c_{j_1}) (\cd_{j_2} + c_{j_2})}\nn \\
    &= \sum_{j=1}^L{i\chi_{2j-1}\chi_{2j}} - \eta \sum_{1 \leq j_1 < j_2 \leq L}{q^{j_2 - j_1 - 1} i\chi_{2j_1} \chi_{2j_2-1}} = \sum_{j=1}^L{i\chi_{2j-1}\chi_{2j}} - \eta \sum_{r = 1}^{L-1}\sum_{k = 1}^{L -r}{q^{r-1} i\chi_{2k} \chi_{2(k+r)-1}}.
\label{eq:QXJW}
\end{align}
Similarly, $\mQ^Y_I(q)$ reads
\begin{align}
    \mathcal{Q}^Y_I(q) 
    &= \sum_{j=1}^L{Z_j} + \eta \sum_{1 \leq j_1 < j_2 \leq L}{q^{-(j_2 - j_1 - 1)}Y_{j_1} \left(\prod_{k = j_1+1}^{j_2 - 1} Z_k\right) Y_{j_2}} \nn \\
    &= \sum_{j=1}^L{(1 - 2 n_j)} - \eta \sum_{1 \leq j_1 < j_2 \leq L}{q^{-(j_2 - j_1 - 1)} (\cd_{j_1} + c_{j_1}) (\cd_{j_2} - c_{j_2})}\nn \\
    &= \sum_{j=1}^L{i\chi_{2j-1}\chi_{2j}} - \eta \sum_{1 \leq j_1 < j_2 \leq L}{q^{-(j_2 - j_1 - 1)} i\chi_{2j_1-1} \chi_{2j_2}} = \sum_{j=1}^L{i\chi_{2j-1}\chi_{2j}} - \eta \sum_{r = 1}^{L-1}\sum_{k = 1}^{L -r}{q^{-(r-1)} i\chi_{2k-1} \chi_{2(k+r)}}.
\label{eq:QYJW}
\end{align}
The expressions for the staggered $\mQ^X_{II}(q)$ and $\mQ^Y_{II}(q)$ can be obtained analogously. 
The single-particle levels can be solved by writing down the corresponding matrices brute-force.
For $\mQ^X_{I}(q)$ we find that there are $(L-1)$-fold degenerate levels with eigenvalues $-q^{-1}$ and $+q^{-1}$, and non-degenerate levels with eigenvalues $-q^{L-1}$, $+q^{L-1}$.
The many-body spectrum can be computed from these as usual, and the  many-body eigenvalues are given by $\{q^{-1}[-L + 2(n + m q^L) + (1-q^L)]\}$ for integer $0 \leq n \leq L-1$ and integer $0 \leq m \leq 1$, a total of $2L$ distinct values, which become $L+1$ distinct eigenvalues in the $q \rightarrow 1$ limit, or $L$ distinct eigenvalues in the $L \rightarrow \infty$ limit if $q < 1$, when the $\mQ^X_I(q)$ operator is quasilocal. 
The eigenvalues of the $\mQ_I^Y(q)$ operator are the same with the replacement $q \rightarrow q^{-1}$. 
\section{Free-Fermion Super-Hamiltonians and Hydrodynamics with Strong Zero Modes}\label{app:szmhydro}
\subsection{General Theory}\label{subsec:freefermsuper}
In this appendix, we discuss the structure of super-Hamiltonians required to derive the hydrodynamics of systems that contain SZMs.
Before we move on to specific examples, we discuss some general features of super-Hamiltonians corresponding to free-fermion terms. Very generally, consider the bond algebra
\begin{equation}
    \mA_{\rm FF} = \lgen \{T_A = \frac{1}{4}\sum_{\alpha\beta}A_{\alpha\beta}\chi_\alpha\chi_\beta\} \rgen,
\label{eq:FFalgebra}
\end{equation}
where $A$ is an antisymmetric Hermitian single-particle matrix (defined with a factor of $\frac{1}{4}$ for later convenience), and $\{\chi_\alpha\}$ are the Majorana fermions, and the bond algebra is defined by specifying a set of such $A$'s, e.g., corresponding to different bond terms.
[We can make additional simplifications in the presence of additional structures/symmetries, e.g., when working with complex fermions with charge conservation; we do not pursue writing these out further here.]
As discussed in Sec.~\ref{subsubsec:Brownian}, super-Hamiltonians defined from such generators are generically \textit{interacting} Hamiltonians on the full doubled Hilbert space $\mH \otimes \mH$, e.g., 
\begin{equation}
    \hmP = \sum\nolimits_{A}{\hmL_{A}^\dagger \hmL_{A}},\;\;\;\;\hmL_A = T_A \otimes \mathds{1} - \mathds{1} \otimes T_A^T.
\label{eq:fullsuperham}
\end{equation}
However, when $T_A$ has a quadratic fermion structure, it is known that $\hmL_{A}$ has a $U(1)$ superoperator symmetry~\cite{lastres2024nonuniversality} in that it preserves the number of Majoranas in any operator string. 
That is, the action of the Liouvillian $\hmL_A$ on an $n$ Majorana operator $\chi_{\gamma_1} \cdots \chi_{\gamma_n}$ returns a linear combination of $n$ Majorana operators, i.e., 
\begin{equation}
    \hmL_A \oket{\chi_{\gamma_1}\chi_{\gamma_2}\cdots \chi_{\gamma_n}} = \frac{1}{4}\sum_{\alpha\beta}{A_{\alpha\beta}\oket{[\chi_\alpha\chi_\beta, \chi_{\gamma_1}\chi_{\gamma_2}\cdots \chi_{\gamma_n}]}} = \sum_\alpha{\sum_{\ell = 1}^{n}{A_{\alpha \gamma_\ell}}\oket{\chi_{\gamma_1} \cdots \chi_{\gamma_{\ell-1}}\chi_{\alpha}\chi_{\gamma_{\ell+1}} \cdots \chi_{\gamma_n}}},
\label{eq:nMajoranacomm}
\end{equation}
where we have used the antisymmetry of $A$ to simplify the expression.
[Note that even though we use notation $\oket{O}$ viewing operators as states in the operator space, we are not using any double space formalism and instead evaluate action of $\mathcal{L}_A$ working directly with operators.]
This superoperator symmetry hence also extends to the super-Hamiltonian $\hmP$, which allows the block-diagonalization of $\hmP$ into superoperators $\hmP_n$ that act only within sectors of operators with $n$ Majoranas for $0 \leq n \leq 2L$.  
The zero energy ground states of $\hmP_n$, if any, are the elements of the commutant of $\mA_{\rm FF}$ of Eq.~(\ref{eq:FFalgebra}). 
We now show that $\hmP_n$ can be interpreted as a single-particle ``hopping problem" on a lattice of $L^n$ sites. 
To start, let us illustrate this for $n = 1$, where from Eq.~(\ref{eq:nMajoranacomm}) we obtain
\begin{equation}
    \hmL_{A} \oket{\chi_\gamma} = \sum_{\alpha}{A_{\alpha\gamma}\oket{\chi_\alpha}},\;\;\implies\;\;\obra{\chi_\beta}\hmP_1 \oket{\chi_\gamma} = \sum\nolimits_{A}{[A^2]_{\beta\gamma}}.
\label{eq:Paction}
\end{equation}
Hence on the space of single Majorana operators, $\hmP$ is a simple positive semi-definite $2L \times 2L$ matrix that can be computed using the generators.
Its zero energy eigenstates are operators in the commutant of $\mA_{\rm FF}$ that are linear combinations of single Majorana operators. 
In this case, note that these are also the Majorana zero modes of each of the $T_A$ itself, which is precisely precisely the form of the SZMs we discussed in this work. 
The same applies to the hydrodynamic modes, which are the low-energy excitations of $\hmP_1$, which are composed of single Majorana operators. 
We will demonstrate these with an explicit example of the Ising SZM in Sec.~\ref{subsec:IsingSZMapp}.
Similar ideas can be applied to the space of $n$ Majorana operators for $n \geq 1$, however the expressions for the effective ``hopping matrix" for $\hmP_n$ can be more involved.  
For example, for $n = 2$, we can use the presentation of operators in terms of antisymmetric matrices as
\begin{equation}
    \oket{\Omega} \defn \frac{1}{4}\sum_{\gamma,\delta}{\Omega_{\gamma, \delta}\oket{\chi_\gamma\chi_\delta}},\;\;\;\Omega^T = -\Omega.
\label{eq:bilinears}
\end{equation}
The actions of $\hmL_A$ and $\hmP$ on these operators read
\begin{equation}
     \hmL_A\oket{\Omega} = \oket{[A, \Omega]}\;\;\;\implies\;\;\;\hmP\oket{\Omega} = \sum_A{\oket{[A, [A, \Omega]]}}.
\label{eq:LPaction2majorana}
\end{equation}
This shows that if $\oket{\Omega}$ is literally viewed as the vectorization of the $2L \times 2L$ matrix $\Omega$, the action of $\hmP$ is equivalent to the action of a ``super-Hamiltonian" of the form 
\begin{equation}
    \hmP_2 = \sum_A{\hml_A^\dagger \hml_A},\;\;\;\hml_A \defn A \otimes \mathds{1} - \mathds{1} \otimes A^T, 
\end{equation}
which can be interpreted as a single-particle ``hopping problem" on the space of $2L \times 2L$ matrices, which has a total of $4L^2$ ``sites" (strictly speaking, we need to further restrict to the subspace of antisymmetric $2L \times 2L$ matrices, i.e., $L(2L-1)$ independent ``sites'').  
Hence the ground states of $\hmP_2$ are the $2L \times 2L$ matrices that commute with all the $A$'s, which can be translated into the set of quadratic Majorana operators that commute with all the $T_A$'s.   
%
%
%
%
%
This can in principle be solved to obtain the exact hydrodynamic modes in many cases, e.g., for the $U(1)$ symmetry demonstrated using the interacting super-Hamiltonian in Ref.~\cite{moudgalya2023symmetries}, as well as for the quasilocal $U(1)$ symmetries discussed in Sec.~\ref{subsec:quasilocalU1I}.
However, we will not illustrate those computations here.
\subsection{Example: Ising SZM}\label{subsec:IsingSZMapp}
We now illustrate the super-Hamiltonians with the Ising SZM as an example, where the bond algebra is generated by the terms $\{\hh_{j,j+1} \defn X_j X_{j+1} + q Z_j = i\chi_{2j}(-q\chi_{2j-1} + \chi_{2j+1})\}$, which have quadratic Majorana forms.
We do this in two complementary ways, first using a variational computation that is agnostic of the free-fermion structures of the generators, and second using the general theory discussed in the previous subsection.
\subsubsection{Variational Calculation}
We start with these generators $\{\hh_{j,j+1}\}$ and variationally study the full super-Hamiltonian of the form of Eq.~(\ref{eq:superhamiltonian}) to demonstrate its ground states and low-energy excitations.
Given any variational ansatz $\Psi(z_k)$, its energy under the super-Hamiltonian can be obtained as
\begin{equation}
	E_k \defn \frac{\obra{\Psi(z_k)} \sum_j{\hmL^\dagger_{j,j+1} \hmL_{j,j+1}} \oket{\Psi(z_k)}}{\obraket{\Psi(z_k)}{\Psi(z_k)}} = \sum_j{ \frac{\obraket{[\Psi(z_k),  \hh_{j,j+1}]}{[\Psi(z_k),  \hh_{j,j+1}]}}{\obraket{\Psi(z_k)}{\Psi(z_k)}}} = \sum_j{\frac{\textrm{Tr}([\Psi(z_k),  \hh_{j,j+1}]^\dagger [\Psi(z_k),  \hh_{j,j+1}])}{\textrm{Tr}(\Psi(z_k)^\dagger \Psi(z_k))}}.
\label{eq:Evar}
\end{equation}
In the case of the Ising SZM, we choose the variational ansatz of the form
\begin{equation}
	\Psi(z_k) \defn \sumal{j = 1}{L}{z_k^{j-1} \left(\prodal{i = 1}{j-1}{Z_i}\right) X_j}.
\label{eq:Psiqk}
\end{equation}
For $z_k = q$,  which simply reduces to the SZM $\Psi^\ell(q)$ of Eq.~(\ref{eq:ISZMdefns}), we obtain an exact ground state of the super-Hamiltonian, and it is straightforward to see that $E_k = 0$ according to Eq.~(\ref{eq:Evar}). 
To obtain the low-energy excitations, we intend to choose $z_k$ to be ``close" to $q$ in order to locally mimic the structure of the exact ground state. 
In order to obtain a useful bound for the gap of $\hmP$,  we should also require that $\oket{\Psi(z_k)}$ to be orthogonal to $\Psi(q)$,  which can be verified to happen only if 
\begin{equation}
\sum_{j = 1}^{L-1}{(z_k^\ast q)^{j-1}} = 1 \;\;\implies\;\; z_k = \frac{1}{q} e^{i \frac{2\pi n}{L}},\;\;\;n \neq 0.
\label{eq:orthocondition}
\end{equation}
If $|q| \neq 1$,  $|z_k - q|$ can no longer be ``small",  hence we do not expect $\Psi(z_k)$ to be a gapless excitation. 
On the other hand,  if $q = \pm 1$,  we can choose $z_k = e^{i k}$ with $k = \frac{2\pi n}{L}$  such that Eq.~(\ref{eq:orthocondition}) is satisfied with $|z_k - q|$ being small (decreasing with system size $L$). 
We hence restrict to $q = 1$ and compute $E_k$ from Eq.~(\ref{eq:Evar}) for $z_k = e^{ik}$.
We obtain the following commutator
\begin{gather}
[\hh_{j,j+1}, \Psi(z_k)] =  \Bigg( \prod_{m = 1}^{j-1}{Z_m} \Bigg) \times {[X_j X_{j+1} + q Z_j,  X_j + z_k Z_j X_{j+1}]} = \Bigg( \prod_{m = 1}^{j-1}{Z_m} \Bigg) \times 2 i (q - z_k) Y_j,
\end{gather}
and using Eq.~(\ref{eq:Evar}) we obtain
\begin{equation}
E_k =  4 \left(1 - \frac{1}{L}\right) \times |q - z_k|^2 = 8 \left(1 - \frac{1}{L}\right) \times (1 - \cos(k)) \sim \frac{1}{L^2}\;\;\text{for large $L$}.
\label{eq:variationalhydro}
\end{equation}
This is an upper bound on the gap of the super-Hamiltonian,  and shows that the super-Hamiltonian for $q = 1$ is gapless in the thermodynamic limit. 
This leads to hydrodynamic modes of the form Eq.~(\ref{eq:Psiqk}) that lead to dynamical signatures discussed in Sec.~\ref{subsubsec:edgecorr}.
\subsubsection{Exact calculation}
The same results can be derived exactly by exploiting the free-fermion nature of the generators $\{\hh_{j,j+1}\}$ and using the general theory of Sec.~\ref{subsec:freefermsuper}. 
When expressed in terms of Majoranas as in Eq.~(\ref{eq:IsingMajorana}), the antisymmetric single-particle matrices $\{A^{j,j+1}\}$ corresponding to the terms $\{\hh_{j,j+1}\}$ have the form 
\begin{equation}
    A^{j,j+1} = 
    \begin{pmatrix}
    0 & 2iq & 0 \\
    -2iq & 0 & 2i \\
    0 & -2i & 0
    \end{pmatrix},
\end{equation}
where we have restricted the full $2L \times 2L$ matrix to the non-zero three-site block spanned by the Majoranas $\{\chi_{2j-1}, \chi_{2j}, \chi_{2j+1}\}$.
We can the obtain the super-Hamiltonian $\hmP_1$ restricted to the space of single-Majorana operators using Eq.~(\ref{eq:Paction}) 
\begin{equation}
    \hmP_1 = 4 q\sum_{j = 1}^{L-1}{\left[q \oket{\chi_{2j-1}}\obra{\chi_{2j-1}} +  q^{-1}\oket{\chi_{2j+1}}\obra{\chi_{2j+1}} - \oket{\chi_{2j-1}}\obra{\chi_{2j+1}} - \oket{\chi_{2j+1}}\obra{\chi_{2j-1}}\right]} + 4(q^2 + 1)\sum_{j = 1}^{L-1}{\oket{\chi_{2j}}\obra{\chi_{2j}}}.
\label{eq:P1IsingSZM}
\end{equation}
The even and odd sites completely decouple, and we obtain $\{\oket{\chi_{2j}}\}$ to be $L-1$ degenerate eigenvectors with eigenvalue $4(q^2 + 1)$.
The last even ``site'', $\oket{\chi_{2L}}$, is completely absent in the above super-Hamiltonian and is hence an eigenvector with eigenvalue $0$; this corresponds to the $i Q_Z X_L$ of Eq.~(\ref{eq:rightmajorana}) in the commutant.
The remaining $L$ eigenvectors, i.e., all odd ``sites'', can be solved by interpreted this matrix as the Hamiltonian of a particle hopping on $L$ sites with open boundary conditions, whose spectrum can be solved exactly using standard structures of tridiagonal Toeplitz matrices with perturbed corners~\cite{toeplitz2008eig}.
Those eigenvalues and unnormalized eigenvectors are
\begin{gather}
    \lambda_0 = 0,\;\;\;\oket{\lambda_0} = \sum_{j = 1}^{L}{q^{j-1}\oket{\chi_{2j-1}}},\nn \\
    \lambda_m = 4\left[q^2 + 1 - 2q\cos(\frac{m \pi}{L})\right],\;\;\oket{\lambda_m} = \sum_{j = 1}^{L}{\left[\sin(\frac{jm\pi}{L}) - q^{-1}\sin(\frac{(j-1)m\pi}{L})\right]}\oket{\chi_{2j-1}},\;\;1 \leq m \leq L-1.
\end{gather}
$\oket{\lambda_0}$ is precisely the Ising SZM $\Psi^\ell(q)$ of Eq.~(\ref{eq:IsingSZMfermion}), and the gap in this sector is given by $4(q - 1)^2$. 
Hence for $q = 1$, $\hmP_1$ is gapless, and the eigenvectors $\oket{\lambda_m}$ for small $m$ are the hydrodynamic modes corresponding to the SZM.
These are also captured by the variational ansatz of Eq.~(\ref{eq:Psiqk}) and have eigenvalues of $E_k$ of Eq.~(\ref{eq:variationalhydro}) for large $L$.
\section{Properties of the Fendley Strong Zero Mode in the  MPO Language}\label{app:fendleysimplified}
In this appendix, we provide additional details about the Strong Zero Modes in the spin-1/2 XYZ model studied in Ref.~\cite{fendley2016strong} that are revealed using the recently discovered MPO expression of Eq.~(\ref{eq:XYZSZMMPO}).
\subsection{Alternate Proof for the Exact SZM}\label{subsec:altproof}
We first provide an alternate proof to existence of this SZM explicitly using the MPO language.
The nature of the proof is similar to the recent one in Ref.~\cite{fendley2025xyz}, with only some technical differences such as the grouping of terms. 
To show that the Fendley SZM indeed commutes with $\tH_{\textrm{XYZ}}$ of Eq.~(\ref{eq:HXYZ}),  we first introduce some convenient notations to work with MPOs.
We denote an MPO with open auxiliary indices by $\opket{\bullet}$, where $\bullet$ is the MPO tensor for some block of sites.
For example, the one-site MPO is denoted by $\opket{F_j}$, and the three-site MPO of Eq.~(\ref{eq:3siteMPO}) is denoted by $\opket{A_j A_{j+1} A_{j+2}}$.
In addition, if the left or right boundary vector is contracted, we represent the tensor with an extra left or right bracket, i.e., as $\opket{[\bullet}$ or $\opket{\bullet]}$ respectively.
We will use both brackets if both the left and right boundary vectors are contracted, i.e., $\opket{[\bullet]}$; in this case the MPO corresponds to a usual operator on the full physical Hilbert space. 
For example, in this notation, the Fendley SZM can be written as 
\begin{equation}
    \Psi(x,y) = \opket{[F_1 F_2 \cdots F_{L}]}.
\label{eq:PsiMPOnotation}
\end{equation}
Consider an MPO tensor $\opket{M}$ whose elements are denoted by $\{M^{[s t]}_{\alpha\beta}\}$, where $\{s, t\}$ and $\{\alpha,\beta\}$ are the physical and auxiliary indices respectively,  and a physical operator $\hO$ with matrix elements $\{O_{st}\}$.
The adjoint action of an operator $\hO$,  i.e., the action of the Liouvillian of the operator $\hO$ on an MPO tensor $\opket{M}$ is then defined as
\begin{equation}
    \mL_{\hO}\opket{M} \defn \opket{\fM},\;\;\;\fM^{[st]}_{\alpha\beta} = \sumal{v}{}{\left(O_{sv}M^{[vt]}_{\alpha\beta} - M^{[sv]}_{\alpha\beta} O_{vt}\right)},
\label{eq:liouvdefn}
\end{equation}
which is simply the commutator of the operator $\hO$ with each physical operator in the MPO.
Defining the following two-site MPOs
\begin{equation}
    \opket{\mX_{j,k}} \defn \mL_{X_j X_k}\opket{F_j F_k},\;\;\; \opket{\mY_{j,k}} \defn \mL_{Y_j Y_k}\opket{F_j F_k},\;\;\;\opket{\mZ_{j,k}} \defn \mL_{Z_j Z_k}\opket{F_j F_k},
\label{eq:twositeMPOs}
\end{equation}
the adjoint action of the three-site term $T_{j,k,l} \defn x X_j X_k + y Y_j Y_k + Z_k Z_l$ that  appears in $\tH_{\textrm{XYZ}}$ (with specific $k=j+1, l=j+2$, but kept general for conciseness) reads
\begin{gather}
    \mL_{T_{j,k,l}}\opket{F_j F_k F_l} = \left(x \mL_{X_j X_k} + y \mL_{Y_j Y_k} + \mL_{Z_k Z_l}\right)\opket{F_j F_k F_l} = x \opket{\mX_{j,k} F_l} + y \opket{\mY_{j,k} F_l} + \opket{F_j \mZ_{k,l}}.
\label{eq:Taction}
\end{gather}
Hence, from Eq.~(\ref{eq:HXYZ}),  the adjoint action of the full Hamiltonian $\tH_{\textrm{XYZ}}$ on the SZM operator $\Psi(x,y) = \opket{[F_1 \cdots F_{L}]}$ is given by
\begin{gather}
   \mL_{\tH_{\textrm{XYZ}}}\opket{[F_1 \cdots F_{L}]} =  \mL_{Z_1 Z_2}\opket{[F_1 \cdots F_{L}]} + \sumal{j = 1}{L-2}{\mL_{T_{j,j+1,j+2}}\opket{[F_1 \cdots F_{L}]}}\nn \\
   = \opket{[\mZ_{1,2} F_3 \cdots ]} + \sumal{j = 1}{L-2}{\left(x\opket{[\cdots F_{j-1} \mX_{j,j+1} F_{j+2}\cdots]} + y \opket{[\cdots F_{j-1} \mY_{j,j+1} F_{j+2} \cdots]} + \opket{[\cdots F_{j} \mZ_{j+1,j+2} F_{j+3} \cdots ]}\right)} .
\label{eq:XYZadjoint}
\end{gather}
While Eq.~(\ref{eq:XYZadjoint}) is not very illuminative,  we remarkably find that there exists a two-site MPO tensor $C_{j,k}$ such that the following properties are satisfied for $j \neq k \neq l$
\begin{gather}
     x \opket{\mX_{j,k} F_l} + y \opket{\mY_{j,k} F_l} + \opket{F_j \mZ_{k,l}} = -\opket{C_{j,k} F_l} + \opket{F_j C_{k,l}},\;\;\;\opket{[\mZ_{j,k}} = \opket{[C_{j,k}} ,\;\;\;\opket{C_{k,l}]} = 0.
\label{eq:Ctensorproperties}
\end{gather}
The existence of such a tensor leads to a telescoping structure that we discuss in the main text as well as below. 
Before that,  we remark that in order to explicitly solve for $C_{j,k}$ with the desired properties,  we can contract the left and right boundary vectors with the first equation and use the second and third equations in Eq.~(\ref{eq:Ctensorproperties}) to obtain 
\begin{equation}
    \opket{[F_j C_{k,l}} = x \opket{[\mX_{j,k} F_l} + y \opket{[\mY_{j,k} F_l} + \opket{[F_j \mZ_{k,l}} + \opket{[\mZ_{j,k} F_l},\;\;\; \opket{C_{j,k} F_l]} = -\left\{x \opket{\mX_{j,k} F_l]} + y \opket{\mY_{j,k} F_l]} + \opket{F_j \mZ_{k,l}]}\right\}.
\label{eq:Ctensorextra}
\end{equation}
The equations in Eq.~(\ref{eq:Ctensorextra}), along with the second and third equations in Eq.~(\ref{eq:Ctensorproperties}) in total form 16 linear equations that completely determine all the elements of the tensor $C_{j,k}$. 
The elements of $C_{j,k}$ read
\begin{equation}
    C_{j,k} = 2 i \sqrt{x y} \times
    \begin{pmatrix}
    0 & \sqrt{x(x^2 - 1)} (y Z_j Y_k + Y_j Z_k) & - \sqrt{y(y^2 - 1)} (x Z_j X_k + X_j Z_k) & 0 \\
    y \sqrt{x (x^2 - 1)} (y Z_j Y_k + Y_j Z_k) & 0 & y \sqrt{(x^2 - 1)(y^2-1)} Z_j & 0 \\
    - x \sqrt{y(y^2 -1)} (x Z_j X_k + X_j Z_k) & - x \sqrt{(x^2 - 1)(y^2-1)} Z_j & 0 & 0 \\
    0 & 0 & 0 & 0
    \end{pmatrix}.
\label{eq:Cjkexplicit}
\end{equation}
Using Eq.~(\ref{eq:Ctensorproperties}), we obtain
\begin{align}
    \mL_{T_{j,j+1,j+2}}\opket{[F_1 \cdots F_{L}]} &= -\opket{[F_1 \cdots F_{j-1} C_{j,j+1} F_{j+2} \cdots F_L]} + \opket{[F_1 \cdots F_j C_{j+1, j+2} F_{j+3} \cdots F_L]}, \nn \\
    \mL_{Z_1 Z_2}\opket{[F_1 \cdots F_{L}]} &= \opket{[C_{1,2} F_3 \cdots F_L]}.
\end{align}
This is precisely Eq.~(\ref{eq:telescopingcomm}) when we define
\begin{equation}
    \widehat{\Psi}_{j,j+1}(x,y) \defn \opket{[F_1 \cdots F_{j-1} C_{j,j+1} F_{j+2} \cdots F_L]},
\label{eq:Psihatdefn}
\end{equation}
and note that $\widehat{\Psi}_{L-1,L} = 0$ due to Eq.~(\ref{eq:Ctensorproperties}).
We can then rewrite Eq.~(\ref{eq:XYZadjoint}) as
\begin{align}
    \mL_{\tH_{\textrm{XYZ}}}\opket{[F_1 \cdots F_{L}]} &= \opket{[C_{1,2} F_3 \cdots ]} + \sumal{j = 1}{L-2}{\left(-\opket{[\cdots F_{j-1} C_{j,j+1} F_{j+2} \cdots]} + \opket{[\cdots F_{j} C_{j+1,j+2} F_{j+3}\cdots]}\right)}\nn \\
    &= \opket{[C_{1,2} F_3 \cdots ]} - \opket{[C_{1,2} F_3 \cdots ]}  + \opket{[ \cdots F_{L-2} C_{L-1,L}]} = 0.
\end{align}
This shows that $\mL_{\tH_{\textrm{XYZ}}}\opket{[F_1 \cdots F_{L}]} = 0$, or in other words $[\tH_{\textrm{XYZ}}, \Psi(x,y)] = 0$.
\subsection{Absence of a Bond Algebra}\label{sec:parentliouvillian}
The MPO form of $\Psi(x,y)$ also allows us to search for a bond algebra with a commutant containing $\Psi(x,y)$.
In particular, we have checked that there is no strictly local term of range 6 or lower that commutes with $\Psi(x,y)$. 
To do so, we look for an operator $\Pi_{[j,j+n]}$ with support over sites $j$ to $j+n$ such that
\begin{equation}
    \mL_{\Pi_{[j,j+n]}}\oket{F_j F_{j+1} \cdots F_{j+n}} = 0 \;\;\iff \;\;[\Pi_{[j,j+n]}, (F_j F_{j+1} \cdots F_{j+n})_{\alpha\beta}] = 0\;\;\forall\; 1\leq \alpha,\beta \leq \chi.
\label{eq:indcommute}
\end{equation}
This means that $\Pi_{[j,j+n]}$ must belong to the commutant of the set of operators $\{(F_j F_{j+1} \cdots F_{j+n})_{\alpha\beta}\}$. 
Doing this explicitly for the MPO $F_j$ of Eq.~(\ref{eq:XYZSZMMPO}) and generic values of $x, y \neq 0$, we find that the only operator in the commutant of the operators $\{(F_j F_{j+1} \cdots F_{j+n})_{\alpha\beta}\}$ is $\mathds{1}$ for $n \leq 6$, which rules out the existence of any non-trivial $\Pi_{[j,j+n]}$ that satisfies Eq.~(\ref{eq:indcommute}).
However, when either $x = 0$ or $y = 0$, there are several strictly local operators that commute with the MPO of Eq.~(\ref{eq:psiMPO}), and we can recover the bond algebras such as Eq.~(\ref{eq:YZSZMbondalgebra}).
\subsection{Normalizability in the Thermodynamic Limit}
We now show that the MPO expression for $\Psi(x,y)$ also reproduces some other calculations in \cite{fendley2016strong}, such as normalizability in the thermodynamic limit. 
The Frobenius norm of the operator $\Psi(x,y)$ is defined as
\begin{equation}
    \norm{\Psi(x,y)}^2 \defn \frac{\Tr(\Psi(x,y)^\dagger \Psi(x,y))}{\Tr(\mathds{1})} = \frac{\Tr(\Psi(x,y)^\dagger \Psi(x,y))}{2^{L}}.
\label{eq:psiLnorm}
\end{equation}
This can be computed from the MPO using an appropriate ``transfer matrix" \begin{equation}
    E \defn \frac{1}{2} \sumal{s, t}{}{F^{[s t]} \otimes \left(F^{[s t]}\right)^\ast},
\label{eq:normtransfer}
\end{equation}
where the $\otimes$ acts on the auxiliary indices, and we have used the fact that physical operators in the tensor $F$ are Hermitian.
Since all the physical indices have been contracted in Eq.~(\ref{eq:normtransfer}), $E$ can be viewed as a $\chi^2 \times \chi^2$ matrix with elements as numbers.
Equation~(\ref{eq:psiLnorm}) can then be expressed as
\begin{equation}
    \norm{\Psi(x,y)}^2 = (b_l^T \otimes b_l^T) E^{L} (b_r \otimes b_r), 
\label{eq:normexp}
\end{equation}
In the following, we will evaluate Eq.~(\ref{eq:normexp}) using the block structure of $E$ under a partition into $15$-dimensional and $1$-dimensional blocks. 
Throughout this discussion, we will abuse notation and use $0$ and $1$ to denote the zero matrix and identity matrices of appropriate dimensions, and we will not explicitly specify their dimensions unless there is some ambiguity.
Using the expression of Eq.~(\ref{eq:XYZSZMMPO}), it is easy to show that $E$ has the following block upper-triangular structure, and hence we have
\begin{equation}
E =  \begin{pmatrix}
   U & V\\
   0 & 1
\label{eq:Estructure}
\end{pmatrix}\;\;\implies\;\;
E^n = 
    \begin{pmatrix}
    U^n & (1 - U)^{-1}(1 - U^n) V\\
    0 & 1
    \end{pmatrix},
\end{equation}
where $V$ is a column vector of size $15$ that has the form $\begin{pmatrix} 1 & 0 & \cdots & 0\end{pmatrix}^T$ and $U$ is a $15 \times 15$ matrix that, when rows and columns are permuted to the order $\{1,6,11,2,5,3,9,4,7,8,10,12,13,14,15\}$ has the block-diagonal form 
\begin{gather}
\Pi U \Pi^\dagger = \operatorname{diag}\!\Big(
\widetilde{U}^{(3)},\ \widetilde{U}^{(2,1)},\ \widetilde{U}^{(2,2)},\ xy,\ xy,\ x,\ xy,\ y,\ xy,\ x,\ y
\Big), 
\label{eq:Uexpr}\\
\widetilde{U}^{(3)}=
\begin{pmatrix}
x^2 y^{2} & |(x^{2}-1)y| & |x(y^{2}-1)|\nonumber \\
|(x^{2}-1)y| & x^{2} & 0\\
|x(y^{2}-1)| & 0 & y^{2}
\end{pmatrix},\;\;
\widetilde{U}^{(2,1)}=
\begin{pmatrix}
x^{2}y & |(x^{2}-1)y|\\
|(x^{2}-1)y| & x^{2}y
\end{pmatrix},\;\;
\widetilde{U}^{(2,2)}=
\begin{pmatrix}
xy^{2} & |x(y^{2}-1)|\\
|x(y^{2}-1)| & xy^{2}
\end{pmatrix},
\end{gather}
where $\Pi$ is the unitary that performs the permutation.
For $|x|, |y| < 1$, it is easy to show that the absolute values of the eigenvalues $\{u_\alpha\}$ of $U$ have eigenvalues strictly less than $1$ by bounding the spectral-radius by the infinity-norm~\cite{horn_johnson_matrix_analysis_2012}, which can then be bounded by  explicitly using the form of the matrix elements.
That is, we have
\begin{equation}
    \rho(U) \defn \textrm{max}_\alpha\{{|u_\alpha|}\} \leq \norm{U}_\infty \defn \max_i\{\sum\nolimits_j{|U_{i,j}|}\} < 1\;\;\;\text{for}\;\; |x|, |y| < 1  
\label{eq:Uineq}
\end{equation}
where $U_{i,j}$ are the matrix elements of $U$, and the rightmost inequality can be shown by explicit computation for each of the blocks in Eq.~(\ref{eq:Uexpr}).
In particular, all of the $\sum_j{|U_{i,j}|}$ can be upper bounded by $1$ using the fact that $|a^2 - 1| = 1 - |a|^2$ when $|a| < 1$, and the inequalities
\begin{gather}
    |a|^2 |b|^2 + |a||(b^2 - 1)| + |b||(a^2 - 1)| = |a|^2 |b|^2 + |b|(1 - |a|^2) + |a|(1 - |b|^2) \leq |a| + |b| - |a||b| = 1 - (1 - |a|)(1 - |b|) < 1, \nn \\
    |a|^2 + |b||(a^2 - 1)| = |a|^2 + |b|(1 - |a|^2) = |b| + |a|^2(1 - |b|) < 1,
    \label{eq:Uineqs}
\end{gather}
which hold for $|a|< 1$ and $|b|< 1$.
Hence in the thermodynamic limit, using Eqs.~(\ref{eq:Estructure}) and (\ref{eq:Uexpr}) we can evaluate Eq.~(\ref{eq:normexp}) as 
\begin{equation}
    \lim_{L \rightarrow \infty} \norm{\Psi(x,y)}^2 = (b_l^T \otimes b_l^T) 
    \begin{pmatrix}
    0 & (1 - U)^{-1} V \\
    0 & 1
    \end{pmatrix} (b_r \otimes b_r) = [(1 - U)^{-1}]_{1,1} = [(1 - \widetilde{U}^{(3)})^{-1}]_{1, 1} = \frac{1}{(1 - x^2)(1-y^2)},
\label{eq:norminf}
\end{equation}
where we have used the structure of $V$, the fact that 
\begin{equation}
b_l^T \otimes b_l^T = \begin{pmatrix} 1 & 0 & \cdots & 0\end{pmatrix},\;\;\;(b_r \otimes b_r) = \begin{pmatrix} 0 & \cdots & 0 & 1\end{pmatrix}^T, 
\label{eq:bl2br2forms}
\end{equation}
and that $U$ has the block-diagonal form as in Eq.~(\ref{eq:Uexpr}).
This recovers the result obtained in Ref.~\cite{fendley2016strong}.
\subsection{Exact to Approximate SZM}\label{subsec:approxSZM}
We can use similar techniques to also show that the norm of the commutator with the full XYZ Hamiltonian is exponentially small in system size, i.e., Eq.~(\ref{eq:Fendleynormexp}).
The norm can be expressed as
\begin{align}
    \mathcal{N}_C &= \norm{[H_{\rm XYZ}, \Psi(x, y)]}^2 = \norm{[x X_{L-1} X_L + y Y_{L-1} Y_L, \Psi(x, y)]}^2 \nn\\
    &= \frac{\text{Tr}([x X_{L-1} X_L + y Y_{L-1} Y_L, \Psi(x, y)]^\dagger [x X_{L-1} X_L + y Y_{L-1} Y_L, \Psi(x, y)])}{2^L}.
\label{eq:Cnorm}
\end{align}
We can evaluate this also using appropriate transfer matrices of the MPO expression for $\Psi(x, y)$. 
Due to this MPO form, we are essentially interested in the commutator of the terms $X_{L-1} X_L$ and $Y_{L-1} Y_L$ with the two-site MPO of $\Psi(x,y)$ on sites $L-1$ and $L$.
In the notation of Sec.~\ref{subsec:altproof}, this can be expressed as an MPO as
\begin{equation}
    [x X_{L-1} X_L + y Y_{L-1} Y_L, \Psi(x, y)] = (x\mL_{X_{L-1} X_L} + y\mL_{Y_{L-1} Y_L})\opket{[F_1 \cdots F_L]} = \opket{[F_1 \cdots F_{L-2} (x \mX + y \mY)_{L-1, L}]},
\label{eq:commMPO}
\end{equation}
where $\mL_{\cdots}$ is the Liouvillian defined in Eq.~(\ref{eq:liouvdefn}), and the MPO tensors $\mX$ and $\mY$ are defined in Eq.~(\ref{eq:twositeMPOs}).
Equation~(\ref{eq:Cnorm}) is hence exactly the norm computation for the MPO of Eq.~(\ref{eq:commMPO}), just as Eq.~(\ref{eq:psiLnorm}) is the norm computation for the MPO of Eq.~(\ref{eq:PsiMPOnotation}).
By analogy to Eq.~(\ref{eq:normexp}), we can similarly write Eq.~(\ref{eq:Cnorm}) as
\begin{equation}
    \mathcal{N}_C = (b_l^T \otimes b_l^T) E^{L-2} \mE (b_r \otimes b_r),
\label{eq:Cnormtrans}
\end{equation}
where we have defined the two-site transfer matrix
\begin{equation}
    \mE\defn \frac{1}{4} \sum_{s,t}{(x \mX + y \mY)^{[st]} \otimes [(x \mX + y \mY)^{[st]}]^\ast}, 
\end{equation}
where $s, t$ are sums over the physical indices of the two-site MPOs $\mX$ and $\mY$.
We can then compute that 
\begin{equation}
    \mE (b_r\otimes b_r) = \begin{pmatrix}
    w_r \\
    0
    \end{pmatrix},\;\;\; \Pi w_r = \begin{pmatrix}
        (w_r^{(3)})^T & 0 & 0 & 0 & 0 & 0 & 0 & 0 & 0 & 0 & 0      
    \end{pmatrix}^T,
\label{eq:mEstruct}
\end{equation}
where the equation on the left is expressing the vector in the block basis of Eq.~(\ref{eq:Estructure}), and the equation on the right is further expressing the non-zero block in the basis of Eq.~(\ref{eq:Uexpr}), with $\Pi$ being the permutation described there and $w_r^{(3)}$ is a $3$-dimensional non-zero column vector.
Using Eqs.~(\ref{eq:Uexpr}), (\ref{eq:mEstruct}), and (\ref{eq:bl2br2forms}), we can write Eq.~(\ref{eq:Cnormtrans}) as
\begin{equation}
    \mathcal{N}_C = [(\widetilde{U}^{(3)})^{L-2} w^{(3)}_{r}]_{1,1}. 
\end{equation}
As shown in Eq.~(\ref{eq:Uineq}), $\widetilde{U}^{(3)}$ has eigenvalues less than $1$ for $|x|, |y| < 1$, and $w^{(3)}_r$ is a vector of bounded norm, hence $\mathcal{N}_C$ decays exponentially in system size. 
A precise computation of the $\mathcal{N}_C$ would require the diagonalization of $\widetilde{U}^{(3)}$, which does not seem to have a very nice expression for general $x$ and $y$.
\subsection{Majorana Nature in the Thermodynamic Limit}
We can also use transfer matrix ideas to show that the operator $\Psi(x,y)^2 \propto \mathds{1}$ in the thermodynamic limit, hence in that limit it is similar to a Majorana Zero Mode (MZM). 
To show this, consider the expansion of $\Psi(x,y)^2$ in the Pauli string basis $\{P^{\mu_1 \cdots \mu_L} \defn \sigma^{\mu_1}_1 \cdots \sigma^{\mu_L}_L\}$, where $\mu_j \in \{0, x, y, z\}$ and $\{\sigma^\mu_j\}$ are the Pauli matrices (with $\sigma^0_j = \mathds{1}_j$), as
\begin{equation}
    \Psi(x,y)^2 = \sum\nolimits_{\{\mu_j\}}{c_{\mu_1 \cdots \mu_L} P^{\mu_1 \cdots \mu_L}},\;\;\;c_{\mu_1 \cdots \mu_L} = \frac{1}{2^L}\textrm{Tr}(P^{\mu_1 \dots \mu_L} \Psi(x,y)^2).
\label{eq:Pauliexpansion}
\end{equation}
We can compute $c_{\mu_1\cdots \mu_L}$ using the MPO expression for $\Psi(x,y)$ as 
\begin{equation}
c_{\mu_1 \cdots \mu_L} = (b_l^T \otimes b_l^T) \Bigg(\prod_{j = 1}^L{E_{\mu_j}} \Bigg) (b_r \otimes b_r),
\label{eq:Pauliweights}
\end{equation}
where the boundary vectors are those of Eq.~(\ref{eq:bl2br2forms}), and $E_{\mu}$ is the generalized MPO transfer matrix defined as
\begin{equation}
    E_{\mu} \defn \frac{1}{2}\sum_{s,t,u}{\sigma^\mu_{st}{F^{[tu]} \otimes F^{[us]}}}.
\label{eq:24transfer}
\end{equation}
Note that $E_0$ is slightly different from $E$ of Eq.~(\ref{eq:normtransfer}) because here we are evaluating $\Psi(x,y)^2$ rather than $\Psi(x,y)^\dagger \Psi(x,y)$, and the MPO representation we are using is not manifestly Hermitian [even though $\Psi(x,y)$ is Hermitian]. 
Nevertheless, many of the general structures of $E$ go through to these transfer matrices $\{E_\mu\}$, e.g., they have the (block) upper triangular forms 
\begin{equation}
    E_0 = 
    \begin{pmatrix}
    U_0 & V_0\\
    0 & 1
    \end{pmatrix},\;\;\;E_z = 
    \begin{pmatrix}
    U_z & V_z \\
    0 & 0
    \end{pmatrix},\;\;\;
    E_\alpha = 
    \begin{pmatrix}
        U_\alpha & 0 \\
        0 & 0
    \end{pmatrix},\;\;\alpha \in \{x,y\}. 
\end{equation} 
Using these, when at least one of the $\{\mu_j \neq 0\}$, the Pauli string weights $\{c_{\mu_1 \cdots \mu_L}\}$ of Eq.~(\ref{eq:Pauliweights}) can be written as
\begin{align}
    c_{\mu_1 \cdots \mu_L} &= \left[\sum_{k = 1}^L{\left(\prod_{j = 1}^{k-1} U_{\mu_j}\right) [V_{0} \delta_{\mu_k,0} + V_z \delta_{\mu_k, z}]} \prod_{j = k+1}^L{\delta_{\mu_j,0}}\right]_{1,1}\nn \\
    &=\left[\left(\prod_{j = 1}^{\ell-1} U_{\mu_j}\right) [U_{\mu_{\ell}} \sum_{k = 0}^{L-\ell-1}{U^k_0} V_{0}  + V_z \delta_{\mu_\ell, z}]\right]_{1,1},\;\;\;\ell = \textrm{max}\{k: \mu_k \neq 0\}\nn \\
    &=\left[\left(\prod_{j = 1}^{\ell-1} U_{\mu_j}\right) [U_{\mu_{\ell}} (1 - U^{L-\ell}_0) (1-U_0)^{-1}V_{0}  + V_z \delta_{\mu_\ell, z}]\right]_{1,1},
\label{eq:wtexpr}
\end{align}
where in the second line we have used the $\delta_{\mu_j,0}$ in the first line to rewrite the sum in terms of $\ell \geq 1$, the position of the rightmost non-trivial Pauli matrix in the string $\sigma^{\mu_1}_1 \cdots \sigma^{\mu_L}_L$.
For convenience of analysis, we can further split Eq.~(\ref{eq:wtexpr}) as
\begin{equation}
    c_{\mu_1 \cdots \mu_L}=\left[\left(\prod_{j = 1}^{\ell-1} U_{\mu_j}\right) [U_{\mu_{\ell}} (1-U_0)^{-1}V_{0}  + V_z \delta_{\mu_\ell, z}]\right]_{1,1} - \left[\left(\prod_{j = 1}^{L} U_{\mu_j}\right) (1-U_0)^{-1}V_{0}\right]_{1,1},
\label{eq:cmudiff}
\end{equation}
where $\ell$ is again the position of the right-most non-trivial Pauli matrix.
We now analyze the two terms in Eq.~(\ref{eq:cmudiff}) separately. 
To analyze the second term in Eq.~(\ref{eq:cmudiff}), we first observe a couple of  properties. 
When permuted to the order $\{1, 2, 3, 5, 6, 7, 9, 10, 11, 4, 8, 12, 13, 14, 15\}$, \textit{all} the $\{U_\mu\}$ can be arranged into the block upper-triangular forms
\begin{equation}
    \Pi U_\mu \Pi^\dagger = 
    \begin{pmatrix}
        A_\mu & R_\mu & S_\mu \\
        0 & T_\mu & 0\\
        0 & 0 & T_\mu
    \end{pmatrix},
\label{eq:Umublock}
\end{equation}
where $\Pi$ is the unitary that implements this permutation, and the dimensions of the sub-matrices are $\{A_\mu\}: 9 \times 9$, $\{R_\mu\}: 9 \times 3$, $\{S_\mu\}: 9 \times 3$, $\{T_\mu\}: 3 \times 3$.
This is similar to the block-diagonalization of $U$ in Eq.~(\ref{eq:Uexpr}), but is coarser to also include $U_{\mu \neq 0}$. 
Next we observe that in this permuted basis  $(1 - U_0)^{-1} V_0$ has the form 
\begin{equation}
    \Pi (1 - U_0)^{-1} V_0 = 
    \begin{pmatrix}
    W_0 & 0 & 0     
    \end{pmatrix}^T,
\end{equation}
where $W_0$ is a 9-dimensional vector. 
This means that the second term in Eq.~(\ref{eq:cmudiff}) can be expressed as
\begin{equation}
    \left[\left(\prod_{j = 1}^{L} U_{\mu_j}\right) (1-U_0)^{-1}V_{0}\right]_{1,1} = \left[\left(\prod_{j = 1}^{L} A_{\mu_j}\right) W_0\right]_{1,1}.
\end{equation}
We can then show that all $\{A_\mu\}$ have infinity-norms bounded as
\begin{equation}
    \norm{A_\mu}_\infty\defn \max_i\{\sum\nolimits_j{(A_\mu)_{i,j}}\} < 1\;\;\;\text{for}\;\;|x|, |y| < 1,
\label{eq:infnormbound}
\end{equation}
which can be shown using the explicit matrix forms and fairly elementary inequalities like those in Eq.~(\ref{eq:Uineqs}).
We then obtain the bound 
\begin{equation}
    \left|\left[\left(\prod_{j = 1}^{L} A_{\mu_j}\right) W_0\right]_{1,1}\right| \leq \prod_{j = 1}^L{\norm{A_{\mu_j}}_\infty} \norm{W_0}_{\infty} \rightarrow 0\;\;\text{as}\;\;L \rightarrow \infty, 
\label{eq:secondtermvanish}
\end{equation}
where we have used Eq.~(\ref{eq:infnormbound}) and the fact that $\norm{W_0}_\infty$ is bounded for any fixed $|x|, |y| < 1$, which can be explicitly proven. 
Finally, we show that the first term in Eq.~(\ref{eq:cmudiff}) vanishes identically for any $\ell \neq 0$, i.e.,
\begin{equation}
    \left[\left(\prod_{j = 1}^{\ell-1} U_{\mu_j}\right) [U_{\mu_{\ell}} (1-U_0)^{-1}V_{0}  + V_z \delta_{\mu_\ell, z}]\right]_{1,1} = 0\;\;\text{for}\;\;\mu_\ell \in \{x, y, z\}. 
\label{eq:firsttermvanish}
\end{equation}
When $\mu_\ell \in \{x, y\}$, this is straightforward to show by direct computation, where we obtain that
\begin{equation}
    U_x (1 - U_0)^{-1} V_0 = 0,\;\;\;\;U_y (1 - U_0)^{-1} V_0 = 0,
\end{equation}
hence causing the first term to vanish. 
The case of $\mu_\ell = z$ is a bit more subtle, but the first term nevertheless vanishes in that case too. 
To see that, we define a vector $v_{xy}$
\begin{gather}
    v_{xy} \defn U_{z}(1-U_0)^{-1} V_0 + V_z = 
    \begin{pmatrix}
    0& 0& 0& 1& 0& 0& i \frac{x y}{fg}& 0& 0& -i \frac{x y}{fg}& 0& 0& 1& 0& 0 
    \end{pmatrix}^T,\nn \\
    f \defn \sqrt{y (x^2 - 1)},\;\;g \defn \sqrt{x (y^2 - 1)}.
\label{eq:vxyexpression}
\end{gather}
Under the action of the $\{U_\mu\}$ matrices, this has the following properties
\begin{equation}
    U_0 v_{xy} = xy v_{xy},\;\;U_x v_{xy} = f v_x,\;\;U_y v_{xy} = g v_y,\;\;U_z v_{xy} = 0,
\label{eq:vxyproperties}
\end{equation}
where $v_x$ and $v_y$ are vectors that read
\begin{align}
    v_x &= 
    \begin{pmatrix}
    0&0&i\frac{x}{fg}&0&0&0&0&1&-i\frac{x}{fg}&0&0&0&0&1&0
    \end{pmatrix}^T, \nn \\
    v_y &= 
    \begin{pmatrix}
    0&-i \frac{y}{fg}& 0&0&i\frac{y}{fg}&0&0&0&0&0&0&1&0&0&1
    \end{pmatrix}^T.
\label{eq:vxvyexpressions}
\end{align}
These vectors $v_x$ and $v_y$ in turn have the following properties under the action of $\{U_\mu\}$:
\begin{align}
    &U_0 v_x = x v_x,\;\;U_x v_x = f v_{xy},\;\;U_y v_x = 0,\;\;U_z v_x = 0,\nn\\
    &U_0 v_y = y v_y,\;\;U_x v_y = 0,\;\;U_y v_y = g v_{xy},\;\;U_z v_y = 0.
\label{eq:vxvyproperties}
\end{align}
From Eqs.~(\ref{eq:vxyproperties}) and (\ref{eq:vxvyproperties}), it is easy to see that the 3-dimensional subspace spanned by $\{v_x, v_y, v_{xy}\}$ is closed under the action of $\{U_\mu\}$. 
Hence in the first term in Eq.~(\ref{eq:cmudiff}) we have a vector in this subspace.
However, from Eqs.~(\ref{eq:vxyexpression}) and (\ref{eq:vxvyexpressions}), the $(1,1)$ elements of all vectors in this subspace is zero, hence the first term of Eq.~(\ref{eq:cmudiff}) is $0$ also if $\mu_\ell = z$. 
Equations~(\ref{eq:secondtermvanish}) and (\ref{eq:firsttermvanish}) imply that the weights $\{c_{\mu_1 \cdots \mu_L}\}$ on all Pauli strings where at least one $\mu_j \neq 0$ vanish, hence combining with the norm of Eq.~(\ref{eq:norminf}), we have
\begin{equation}
    \Psi(x,y)^2 = \frac{1}{(1 - x^2)(1 - y^2)}\mathds{1}\;\;\text{as}\;\;L \rightarrow \infty.
\end{equation}
\end{document}